\documentclass[letterpaper,10pt]{IEEEtran}
\ifCLASSINFOpdf
\else
\fi
\usepackage{amsthm}
\usepackage{amsmath}
\usepackage{graphicx}
\usepackage{multicol}
\usepackage{comment}
\usepackage{multirow} 
\usepackage{epstopdf}
\usepackage{color}
\usepackage{diagbox}
\usepackage{latexsym} 
\usepackage{times}
\usepackage{graphicx}
\usepackage{subfigure}
\usepackage{booktabs}
\usepackage{float}
\usepackage{mathtools, amsfonts, amssymb, amsthm,color,bm} 
\usepackage{cite}
 \usepackage[english]{babel}
\usepackage[figurename=Fig.,font={small,it}]{caption} 
\usepackage{siunitx}
\usepackage{algorithm}
\usepackage{lettrine}
\usepackage{algorithmic} 

\theoremstyle{definition}
\newtheorem{theorem}{\textbf{Theorem}}

\newtheorem{proposition}[theorem]{\textbf{Proposition}}

\usepackage[utf8]{inputenc}
\usepackage[T1]{fontenc}
\usepackage{authblk}
\date{%
}   
 
 \usepackage{xcolor}

\begin{document}
%\title{Real-time Energy Disaggregation at Substation Level via   Dictionary Learning with Partial Labels}
%\title{Uncertainty Bayesian Modeling for Energy Disaggregation with Behind-the-Meter Solar Generation}

\title{
{Perturbed Decision-Focused Learning \\ for Modeling Strategic Energy Storage}}

\author[ ]{
	\IEEEauthorblockN{  Ming Yi, \textit{Member, IEEE,} Saud Alghumayjan, \textit{Student Member, IEEE}, Bolun Xu, \textit{Member, IEEE,}  \thanks{% Paper no.TPWRS-00384-2020. 
	This work was supported in part by the  Data Science Institute in Columbia University and in part by the National Science Foundation under award ECCS-2239046.  (Corresponding author: Ming Yi.)}	\thanks{ Ming~Yi is with the Data Science Institute, Columbia University, New York City, NY. Saud Alghumayjan, and Bolun Xu  are with Earth and Environmental Engineering, Columbia University, New York City, NY. Email:\{my2826, saa2244, bx2177\}@columbia.edu.}
}}
%E. Farantatos and T. Barik are with Electric Power Research Institute, Palo Alto, CA.
%\{efarantatos, TBarik\}@epri.com,
%} 	 } 
%	 }   
 
%Evangelos Farantatos, \textit{Senior Member, IEEE,} Tapas Barik, \textit{Member, IEEE}

%\affil[1]{ Department of Electrical, Computer, and Systems Engineering, Rensselaer Polytechnic Institute, Troy, NY. }
%\affil[2]{ GEIRI North America, 250 W Tasman Dr, San Jose, CA, 95134.} 

\maketitle 
%With the rapid increase in grid-scale energy storage deployment, it has become crucial for power system operators to accurately predict strategic energy storage behaviors in electricity markets, specifically the timings of charging and discharging. This prediction task is complex due to the intricate interplay of market fluctuations and decision-making processes, constrained by physical power and energy limits. 

\begin{abstract}
This paper presents a novel decision-focused framework integrating the physical energy {storage model} into machine learning pipelines.  Motivated by the model predictive control for energy storage,  our end-to-end method incorporates the prior knowledge of the {storage model} and infers the hidden reward that incentivizes energy storage decisions. This is achieved through a dual-layer framework, combining a prediction layer with an optimization layer. We introduce the perturbation idea into the designed decision-focused loss function to ensure the differentiability over linear {storage model}s, supported by a theoretical analysis of the perturbed loss function.  We also develop a hybrid loss function for effective model training.  We provide two challenging applications for our proposed framework: energy storage arbitrage, and energy storage behavior prediction. The numerical experiments on real price data demonstrate that our arbitrage approach achieves the highest profit against existing methods.
The numerical experiments on synthetic and real-world energy storage data show that our approach achieves the best behavior prediction performance against existing benchmark methods, which shows the effectiveness of our method.

\end{abstract}

\section{Introduction}

% Energy storage is crucial to advancing the integration of renewable generation by storing the excess energy of renewable generations and moving the supply to periods when it is scarce, e.g., solar generation is low during the nighttime. The past decades have witnessed a remarkable surge of deployment of energy storage in energy systems. The U.S. Energy Information Administration (EIA) reported that the total battery storage capacity in the United States has climbed to 16 GW by the end of 2023 \cite{EIA23}, and it is expected that the total capacity will expand to more than 30 GW by the end of 2024. Facilitated by Federal Energy Regulatory Commission (FERC) Order 841 \cite{FNRC18}, energy storage can participate in all wholesale energy markets in the United States, and price arbitrage has become the most popular service~\cite{zheng2023energy}. On the other hand, distributed energy storage systems also offer increased demand-side flexibility and are becoming key resources participating in demand response programs~\cite{oldewurtel2013framework}.   

{Over the past decade, energy storage integration has proven essential for economical and reliable power system decarbonization~\cite{EIA23}. Facilitated by Federal Energy Regulatory Commission (FERC) Order 841 \cite{FNRC18}, energy storage systems can now participate in all wholesale energy markets in the United States. Given their limited capacity, energy storage resources must strategically plan their responses based on future price expectations~\cite{harvey2001market}, where this process is known as energy storage arbitrage. Price arbitrage has become one of the most popular services~\cite{zheng2023energy}, whereby energy storage systems discharge when prices are high and charge when prices are low. The energy storage arbitrage provides flexible generation and demand capacity, improving power system efficiency and supporting a more  resilient grid.}

{In practice, most energy storage owners adopt a model predictive control (MPC) framework for storage arbitrage, which decouples price prediction from storage control. They develop proprietary electricity price predictors to capture market fluctuations and input the predicted prices into arbitrage optimization models. However, decoupling prediction and optimization in MPC frameworks introduces storage control and regulation challenges. For storage operators, machine learning predictors often struggle to effectively account for how their predictions will be utilized in the physical energy storage model, particularly given the limited data availability and computing power in practical implementations~\cite{sioshansi2021energy}.}

{Moreover, as each energy storage owner develops their own price predictor to maximize profits, energy storage behaviors are intertwined with electricity pricing dynamics. Given the volatile nature of electricity prices, these arbitrage behaviors exhibit inherent variability, adding uncertainty to the power system and posing significant challenges for grid operation and planning. While energy storage offers promising benefits, its unpredictable actions can disrupt grid stability. Electricity market regulators are increasingly concerned that non-transparent storage arbitrage models may enable operators to exercise market power~\cite{zhou2024energy}, potentially undermining social welfare. As a result, regulators seek systematic approaches to interpret proprietary prediction models and understand the relationship between strategic storage behaviors and market prices. Accurate prediction of energy storage behavior is thus crucial for power system planning and operations, including net-load forecasting, demand response estimation, market power monitoring, and resource adequacy planning.}

{To address these challenges,} this paper proposes a decision-focused framework incorporating the energy {storage behavior} as a perturbed differentiable layer in the learning pipeline. Our proposed framework offers a significant advantage over existing decision-focused approaches \cite{SXL22} in that it does not rely on ground truth prediction information. The result learning model is more efficient at both predicting and optimizing storage operation from the perspective of storage operators or regulators. {The proposed framework is validated in two strategic applications: (1) energy storage arbitrage, which maximizes profit through optimized storage operations, and (2) energy storage behavior prediction, which forecasts storage actions based on historical data.}
The contributions of this paper include: 
\begin{itemize}
  
    \item We develop a decision-focused, end-to-end pipeline incorporating the physical energy {storage model}. The proposed framework includes a prediction layer to infer the hidden reward and an optimization layer to model the decision-making of energy storage. 

    \item We exploit the perturbation idea to solve the differentiable issue in the decision-focused loss function. We also add a predictor regularizer in the loss function to enhance the prediction performance. 
    
    \item We provide a theoretical analysis for the perturbed loss function. Specifically, we prove that the perturbed loss function is convex and that its gradient is Lipschitz continuous. We also show the connection between the perturbed loss function and the original loss function.
    
    \item We validate the proposed framework through two applications:  energy storage arbitrage and forecasting the energy storage behavior. To the best of our knowledge,  this specific formulation for both applications has not been studied before.

    \item The numerical experiments for two applications on synthetic datasets and real-world datasets demonstrate that our approach outperforms the benchmark methods. 
\end{itemize}

The rest of the paper is organized as follows. Section II reviews the related works. Section III introduces the formulation. Section IV presents our proposed end-to-end approach. The numerical results are reported in Section V, and  Section VI concludes our paper.

\section{Motivation and Related Works}
\subsection{Learning-aided Storage Operation}\label{price_pred}

Storage owners can arbitrage electricity prices by charging energy when prices are low and selling it back when prices are high. They can submit bids in the day-ahead market (DAM) and the real-time market (RTM), each with its own market clearing timelines and rules. Various methods have been proposed, such as bi-level optimization~\cite{WDF17} and dynamic programming~\cite{JP152}~\cite{BZX23}. For small-scale energy storage, such as Behind-the-Meter (BTM) energy storage, owners can self-schedule operations without informing the system operator. Many approaches have been proposed, including dynamical programming~\cite{JP15}, model predictive control (MPC)~\cite{ADJ16}\cite{CZZ17}, and reinforcement learning (RL)~\cite{WZ18}. Electricity price fluctuations are the most critical factor in deciding the best bid or schedule options. Various models have been developed over the years to address this uncertainty, employing techniques such as Markov process models~\cite{ZJX22}, frequency analysis~\cite{ZFG22}, and deep neural networks~\cite{LQL23, LWZ22}. {In dynamic programming or MPC setups, the prediction model is trained to minimize forecast error, not decision error, resulting in suboptimal profit maximization. Our proposed decision-focused framework incorporates energy storage decision-making into the model training (i.e., training to minimize decision loss) and enhances arbitrage performance. Additionally, our method avoids the inefficiencies of reinforcement learning (RL),  RL relies on trial-and-error exploration to optimize policies, which can be time-consuming and inefficient for tasks with well-defined objectives, such as energy storage arbitrage. Our method directly measures how well predictions support optimal decisions, providing a more effective solution. Furthermore, RL cannot handle optimization constraints effectively.  RL requires discretizing the action space, and this discretization leads to inefficiencies in representing system dynamics, such as state-of-charge (SoC) transitions, which increase training complexity and reduce interpretability. In contrast, our approach explicitly integrates physical constraints into the model, enabling optimized decision-making and enhancing interpretability.}

%Machine learning-based approaches for bid design include direct bid formulation or a two-stage process involving price prediction and bid generation. \cite{CHL20} uses deep reinforcement learning (DRL) to find the optimal policy for maximizing arbitrage profits, while \cite{BZX23} proposes a two-stage model to forecast opportunity value before bid design. However, these models solely output bid strategies without considering price retrieval, while price prediction methodologies leverage forecasts to design bids through techniques like MPC \cite{AA11}. \cite{CZZ17} proposes an electricity price forecasting framework for self-scheduling in the real-time market. 

%Recent publications have also focused on objective-oriented price prediction. In these models, the objective is not to obtain the best price prediction measured by square errors, but is oriented to the targeted applications. These include demand predictions using custom loss functions based on supply curves~\cite{ZWH22}, predicting real-time prices using a graphical neural network reflecting power flow topology~\cite{liu2021graph}, and recovering offer prices using inverse optimization~\cite{liang2023data}. These works demonstrated the versatility of adapting learning and prediction models over the application while also highlighting the importance of utilizing differentiable models to easily calculate back-propagate gradients, as many power system applications are based on piece-wise linear cost functions.

\subsection{Learning-aided Storage Monitoring}

There is increasing interest in understanding and modeling how energy-limited resources, including storage and demand-side resources, would behave in electricity markets.
One line of related works is modeling the demand response behavior. Both model-based and data-driven approaches are proposed to model the price-responsive behaviors. The model-based approaches  \cite{FMP21} \cite{K21} first solve a bi-level optimization problem to identify the model parameter and then forecast the future demand response behaviors of buildings. For data-driven approaches, reference \cite{NJ17} employs the Gaussian process to model the demand response of a building with energy storage. \cite{VR19} builds a neural network to model customer's demand response behaviors. 

The most relevant research is the energy storage model identification \cite{BZZ22} \cite{BZZ23}, which aims to identify the unknown parameters of the energy {storage model} and uses a pre-built price predictor to forecast the energy storage behavior based on the identified model. On the one hand, knowing the actual price predictor the energy storage owner uses is difficult. Instead, the system operator usually can only observe the historical energy storage behaviors.  On the other hand, the price predictor mainly focuses on predicting the mean value of the price trend but falls short of capturing the price variations. The behaviors of energy storage arbitrage are based on the price spread instead of the true price value.  Therefore, simply feeding the price prediction into the optimization model may not predict the energy storage behavior well. {Compared to the energy storage model identification approaches in \cite{BZZ22} and \cite{BZZ23}, our method can infer dynamics of price forecasting
 directly from historical energy storage decisions. It then uses the predicted reward to estimate future energy storage behavior. Alternatively, one could directly build a neural network to map input features to observed energy storage behaviors. However, such approaches fail to account for the optimization structure of energy storage, leading to degraded prediction performance. }

\subsection{Decision-Focused Learning}

% Many real-world decision-making scenarios need to take future uncertainty into account. The conventional approach typically follows a two-stage process: a prediction model is first trained, and then its output is fed into an optimization model to obtain solutions. However, the prediction model often fails to leverage the structural information by the optimization constraints and objective functions. Moreover, prediction errors are insufficient in accounting for decision-making errors.

%\ The decision-focused learning approach embedded one optimization layer into the machine learning pipeline to effectively account for how their predictions will be used in the physical
%energy {storage behavior}

Decision-focused learning has gained increasing interest in overcoming the limitations of the MPC-type methods in which the prediction model often fails to leverage the structural information by the optimization constraints and objective functions.
One approach is a smart predict-then-optimize framework that integrates optimization into the prediction model and develops a decision-focused surrogate loss function\cite{EG22,ELM20}. This line of research requires the ground truth prediction information for the model training.  Another promising line of research is the learning-by-experience approach~\cite{BJK24,BBT20}. Without costly labeled ground truth prediction data, this method only requires optimal decision information and shows its efficacy in addressing vehicle routing problems.  

Energy {storage model}s have mixed cost terms, including linear and quadratic costs, along with charging/discharging dynamics governed by state-of-charge evolution. The aforementioned methods mainly focus on linear cost functions and lack consideration for the unique characteristics of energy storage systems. While existing work \cite{SXL22} extends the smart predict-then-optimize paradigm to energy storage arbitrage problems, they do not consider degradation costs and depend on perfect forecasting pricing data for model training. Our approach distinguishes \cite{SXL22} by only requiring energy storage decisions rather than relying on ground truth prediction data for training model. This is particular helpful for the scenarios where the ground truth prediction  data is not available. Therefore, we broaden the applicability of decision-focused learning, making it adaptable to various scenarios and offering a more comprehensive solution for strategic energy storage applications.

\section{Problem Formulation}
We first introduce the storage model and the baseline MPC framework, then formulate the learning problem.
% This paper introduces a decision-focused framework and provides two application scenarios for energy storage arbitrage and predicting energy storage behaviors.  The objective of energy storage arbitrage is to maximize total profit based on future price expectations.  The objective of predicting the energy storage arbitrage behaviors is to predict the timing of the storage charges or discharges, according to historical observed market prices and the storage's past operation profiles. We first outline the energy storage arbitrage decision-making model and then provide a general problem formulation. Then we explain how to fit these two applications into our proposed framework.

\subsection{Energy Storage Arbitrage Model}

We consider a model-predictive control (MPC) type storage arbitrage controller in which the storage operator uses a pre-trained model $g(\bm{x})$ to generate price predictions using  historical features $\bm{x}$, and then the storage model makes the optimal arbitrage decisions by solving \eqref{eq:arb}.

\begin{subequations}
\label{eq:arb}
\begin{align}
    (\bar{p}_t, \bar{b}_t) \in \underset{{p_t, b_t}\in \mathcal{X}}{\textbf{argmax}} \; & \sum_{t=1}^T {\lambda}_{t} (p_t-b_t)-u(p_t,b_t)\,,\; \bm{{\lambda}} = g(\bm{x})\label{eq:1a}\\
    \textrm{s.t.} \; & 0\leq  p_t,b_t \leq P \label{eq:1b} \\
    &p_t = 0 \quad \text{if} \quad \lambda_t <0 \label{eq:1c} \\
    &0\leq  e_t \leq E  \label{eq:1d} \\
    & e_t-e_{t-1} =-p_t/\eta +b_t \eta   \label{eq:1e}
\end{align}
\end{subequations}
\noindent where  $p_t$ and $b_t$ are discharge and charge power at time $t$, while $\bar{p}_t$ and $\bar{b}_t$ are the optimized storage dispatch decisions. $T$ is the time horizon. $\lambda_t$ is the predicted real-time price. $e_t$ is the state-of-charge (SoC). $P$ is the power rating of storage. $\eta$ is the storage efficiency. $E$ is the SoC capacity. $u(p_t,b_t)$ is the charge/discharge cost term and it is a convex function over $p_t, b_t$. For example, $u(p_t,b_t)=C_1p_t+C_2p_t^2+C_3b_t+C_4b_t^2$ models the linear and quadratic discharging cost. The objective of the arbitrage model is to maximize its profit while keeping the charge/discharge decision within the physical constraints. \eqref{eq:1b} sets the upper and lower bounds for the charge and discharge power.  {As noted in \cite{XKB20}, constraint \eqref{eq:1c} is a convex relaxation of the non-convex constraint $p_t b_t =0$ that prevents simultaneous charging and discharging in energy storage. The non-convex constraint can also be convexified with a small error, as shown in \cite{SRL17}}.   \eqref{eq:1d} sets the upper and lower bounds for SoC. \eqref{eq:1e} models the SoC evolution process. In the following, we denote the constraints in \eqref{eq:1b}-\eqref{eq:1e} as the feasibility set $\mathcal{X}$.

\subsection{Problem Statement}

Our algorithm aims to infer the hidden reward $\hat{\bm{{\lambda}}}$ that steers the storage operation to achieve the observed energy storage strategy (charge/discharge decisions).  Subsequently, we use this predicted reward to imitate the observed strategy in future energy storage decisions. In particular, our goal is to establish an end-to-end mapping for predicting forthcoming energy storage decisions.    The model parameters are iteratively updated by minimizing a defined loss function,
\begin{subequations} \label{problem_formu}
\begin{align}
& \underset{\bm{\theta}}{\textbf{min}} \, \mathcal{L}\Big(g(\bm{x}),(\bar{p}_t, \bar{b}_t)\Big)\\
&\hat{p}_t,\hat{b}_t=h_\theta(\hat{\bm{x}}|\bm{\theta})
\end{align}
\end{subequations}
\noindent where ${\bm{x}}$ denotes the historical features for training the learning model. $\hat{\bm{x}}$ denotes the input features for predicting the storage behavior. The pair $(\bar{p}_t, \bar{b}_t)$ denotes the observed storage decisions. The set $\bm{\theta}$ contains all the parameters in the proposed end-to-end framework. $\mathcal{L}$ denotes the loss function for learning the mapping. 

\textit{Remark 1. Energy Storage Arbitrage.}
For the energy storage arbitrage, our algorithm’s objective is to infer the hidden reward and imitate these optimal strategies for future arbitrage decisions. We use historical electricity prices to generate optimal arbitrage decisions by solving \eqref{eq:arb}, then use the optimal decisions as training inputs.

\textit{Remark 2. Energy Storage Behavior Prediction.}
For the energy storage behavior prediction, we take the perspective of a regulator to understand and predict the charge and discharge pattern of a particular storage unit given the observed past storage actions $\bar{p}_t, \bar{b}_t$ and the historical features (prices) $\bm{{x}}$. 
% We assume the regulator has knowledge of the storage operating parameters. The regulator can either acquire this information through resource integration registration or use data-driven approaches to estimate~\cite{BZZ23}. On the other hand, 
We do not assume knowledge of future price predictions $\bm{{\lambda}}$ that the storage used to optimize its operation. % Instead, the regulator has a reasonable estimation of the training features $\bm{{x}}$ the storage used to conduct its price prediction.

Our algorithm aims to infer the hidden reward $\hat{\bm{{\lambda}}}$ that steers the storage operation from the historical energy storage decisions and then uses the predicted reward to obtain the future energy storage behavior prediction.   In particular, our objective is to establish an end-to-end mapping for predicting future energy storage decisions.

\section{Methodology}

\begin{figure*}[!ht]
  \centering
  \includegraphics[width = 1\linewidth]{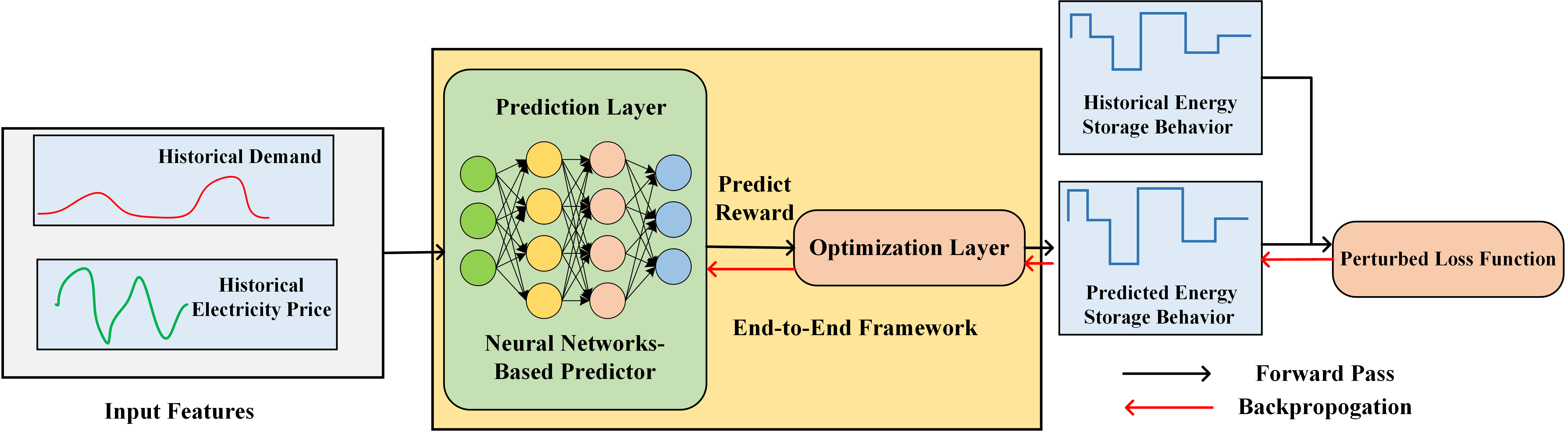} 
    	\caption{{The pipeline of proposed decision-focused prediction approach. Given the input features, the neural network-based predictor first predicts the hidden reward. Subsequently, the optimization layer utilizes this hidden reward to calculate the decision by solving an optimization problem. The algorithm then conducts backward propagation to update the weights in the predictor, based on a perturbed decision-focused loss function.}} \label{Fig: flowchart}
\end{figure*}

Inspired by the MPC control design used in real-world storage arbitrage controls, our framework consists of two components: the prediction layer and the optimization layer. The prediction layer takes the historical feature as the input and outputs the reward. The pipeline embeds the optimization layer as an additional layer. Therefore, this additional layer must support the forward and backward pass to make it compatible with the prediction layer. 
The idea of our proposed approach is illustrated in Fig. \ref{Fig: flowchart}. 
In the following, we will first detail the prediction layer and the optimization layer and then introduce the learning process.

\subsection{The Prediction Layer and the Optimization Layer}

Given the input feature $\bm{x}\in \mathbb{R}^{T\times l}$, where $T$ is the length of time series, and $l$ is the number of different types of features. We build the prediction layer based on deep learning neural networks; this layer's output is the reward prediction. The neural network is parameterized by the weight $\bm{w}$,
\begin{equation}
    \hat{\bm{\lambda}}=g_w(\bm{x})
\end{equation}
where $\hat{\bm{\lambda}} =[\hat{{\lambda}}_1, \hat{{\lambda}}_2,...,\hat{{\lambda}}_T]^T \in \mathbb{R}^{T}$ denotes all rewards $\hat{\lambda}_t$ in time horizon $T$.
The $g_w(.)$ can be any differentiable predictor. In our experiments, we adopt a neural network model as our predictor. As energy storage behaviors strongly correlate with electricity prices and demand, we can include the previous days' day-ahead price (DAP), real-time price (RTP), and demand data as the input features. 

After the hidden reward $\hat{\bm{\lambda}}$ is predicted, the formulation in equation (\ref{eq:arb}) is employed to obtain the storage decisions. The optimization layer $f$ is therefore defined as:
\begin{equation}\label{obj1}
\begin{aligned}
f: \hat{\bm{\lambda}} \mapsto \underset{{p_t, b_t \in \mathcal{X}}}{\textbf{argmax}} & \{ \sum_{t=1}^T \hat{\lambda}_{t} (p_t-b_t)-u(p_t,b_t) \}\\
\end{aligned}
\end{equation}

To simplify the notation, we add additional constraints $y_t=p_t-b_t$ in $\mathcal{X}$ and consider a more general model,
\begin{equation}
\begin{aligned}
f: \hat{\bm{\lambda}} \mapsto \underset{{y_t \in \mathcal{X}}}{\textbf{argmax}} & \{\sum_{t=1}^T \hat{\lambda}_{t} y_t-u(y_t) \}= \underset{{\bm{y} \in \mathcal{X}}}{\textbf{argmax}} \, \{\hat{\bm{\lambda}}^{T} \bm{y}-u(\bm{y}) \}\\
\end{aligned}
\end{equation}
where $\bm{y}=[y_1,y_2,...,y_T]^T \in \mathbb{R}^T$ represents the charge and discharge decisions $(p_t, b_t), t= 1,2,...,T$. $u(\bm{y})$ denotes the total cost: $u(\bm{y})=\sum_{t=1}^Tu(y_t)$. Specifically, the maximization problem is defined as $F(\hat{\bm{\lambda}})=\underset{{\bm{y} \in \mathcal{X}}}{\textbf{max}} \,  \{\hat{\bm{\lambda}}^T \bm{y}-u(\bm{y}) \}$ and the corresponding optimal solution is $\bm{y}^*(\hat{\bm{\lambda}})=\underset{{\bm{y} \in \mathcal{X}}}{\textbf{argmax}} \,  \{\hat{\bm{\lambda}}^T \bm{y}-u(\bm{y})\}$. 
Then the end-to-end policy which maps the input features to the energy storage behaviors is defined as:
\begin{equation}\label{eq:mapping_nn}
\bm{x} \mapsto \bm{y} = f(g_w(\bm{x}))
\end{equation}

\subsection{ The Learning Approach}
To ensure good performance with the proposed algorithm,  it is crucial to compute the weights of the predictor in a manner that aligns the prediction $\bm{y}$ close to the target energy storage decision $\bar{\bm{y}}$. The learning approach involves training a decision-focused end-to-end framework to follow the decisions taken by an unknown policy.

\textbf{Dataset.}
Our learning problem is in a supervised setting and we build a dataset as follows, 
\begin{equation}
\mathcal{D}=\{(\bm{x}^1,\bar{\bm{y}}^1), (\bm{x}^2, \bar{\bm{y}}^2),...,(\bm{x}^n, \bar{\bm{y}}^n)\}
\end{equation}
where $\bm{x}$ denotes the input features,  and $\bm{x}$ includes historical price and/or demand data.  $\bar{\bm{y}}$ denotes the target energy storage decisions. $n$ is the number of data samples in the dataset. For the energy storage arbitrage problem, $\{\bar{\bm{y}}^1,\bar{\bm{y}}^2,...,\bar{\bm{y}}^n \}$ represents the optimal arbitrage decisions generated from historical real-time prices.   For the behavior prediction problem, $\{\bar{\bm{y}}^1,\bar{\bm{y}}^2,...,\bar{\bm{y}}^n \}$ represents the historical observed energy storage behaviors.  

%{As $\bm{\lambda}$ is one feature in $\bm{x}$ and given historical real-time electricity price $\{\bm{\lambda}^1,\bm{\lambda}^2,...,\bm{\lambda}^n\}$ in ${\mathcal{D}}$, the optimal energy storage arbitrage decisions $\{\bar{\bm{y}}^1,\bar{\bm{y}}^2,...,\bar{\bm{y}}^n \}$} can be generated by solving the equation \eqref{eq:arb}. 

In order to achieve optimal performance of the prediction of energy storage decisions, it is essential to design an efficient loss function to update the weights $\bm{w}$ of the neural network. This ensures that the prediction of energy storage decisions aligns closely with the target storage decisions.

\textbf{Learning problem.} The learning problem is to find the optimal $\bm{w}^*$ such that the loss function $\mathcal{L}(.)$ is minimized over the training dataset $\mathcal{D}$.
\begin{equation}
\begin{aligned}
\bm{w}^* =\underset{\bm{w}}{\textbf{argmin}} & \sum_{i=1}^n \mathcal{L}(g_w(\bm{x}^i),\bar{\bm{y}}^i)\\
\end{aligned}
\end{equation}
To ensure the prediction follows the target energy storage decisions, it is natural to define a decision-focused loss function such that the computed charge/discharge decisions result in the same profit as the target decisions. Then given estimated reward $\hat{\bm{\lambda}}$, the target charge and discharge decisions $\bar{\bm{y}}$, the decision-focused loss function is defined as
\begin{equation}\label{eq:loss11}
\begin{aligned}
\mathcal{L}^{FY}(\hat{\bm{\lambda}},\bar{\bm{y}})  &=
 \underset{{\bm{y} \in \mathcal{X}}}{\textbf{max}} \,  \{\hat{\bm{\lambda}}^T \bm{y}-u(\bm{y}) \}- ( \hat{\bm{\lambda}} \bar{\bm{y}}
 -u(\bar{\bm{y}}))\\
  &= \{\hat{\bm{\lambda}}^T \bm{y}^*(\hat{\bm{\lambda}})-u(\bm{y}^*(\hat{\bm{\lambda}})) \}- ( \hat{\bm{\lambda}} \bar{\bm{y}}
 -u(\bar{\bm{y}}))\\
\end{aligned}
\end{equation}
where $\bar{\bm{y}}$ represents the set of target charge and discharge decisions $y_t, t= 1,2,...,T$. {$\hat{\bm{\lambda}}^T$ denotes the transpose of  $\hat{\bm{\lambda}}$.} It is clear that $\mathcal{L}^{FY}(\hat{\bm{\lambda}},\bar{\bm{y}}) \geq 0$ as $\bm{y}^*(\hat{\bm{\lambda}})$ is the maximizer  given $\hat{\bm{\lambda}}$. The equality $\mathcal{L}^{FY}(\hat{\bm{\lambda}},\bar{\bm{y}})= 0$ holds only when $\bm{y}^*(\hat{\bm{\lambda}})=\bar{\bm{y}}$. Unfortunately, the loss function in \eqref{eq:loss11} is not continuous because  $\bm{y}^*(\hat{\bm{\lambda}})$ may not be continuous in $\hat{\bm{\lambda}}$\cite{EG22}\cite{BBT20}.   The discontinuity poses computation difficulties when applying backpropagation to update the weights. We employ the perturbation idea in reference \cite{BBT20} by perturbing the reward $\hat{\bm{\lambda}}$ with additive noise. Then the perturbed  optimizer is:
%To see this, the feasible solutions of $\bm{y}(\hat{\bm{\lambda}})$ form a polytope. The optimal solution  $\bm{y}^*(\hat{\bm{\lambda}})$ is a vertex of the polytope. Given varying $\hat{\bm{\lambda}}$, $\bm{y}^*(\hat{\bm{\lambda}})$ will remain unchanged or jump into another vertex. This discontinuous behavior characterizes  $\bm{y}^*(\hat{\bm{\lambda}})$ as a piece-wise constant function of $\hat{\bm{\lambda}}$.
\begin{equation}
\begin{aligned}
f_\epsilon: \hat{\bm{\lambda}} \rightarrow \mathbb{E}[\underset{{\bm{y} \in \mathcal{X}}}{\textbf{argmax}} & (\hat{\bm{\lambda}} +\epsilon \bm{Z})^T\bm{y}-u(\bm{y})]\\
\end{aligned}
\end{equation}

\noindent where $\bm{Z}$ is additive Gaussian noise $\mathcal{N}(\bm{0},\bm{I})$.  $\epsilon$ is a positive scaling parameter. $\mathbb{E}[.]$ denotes taking the expectation over the variable $\bm{Z}$.  The perturbed maximization problem is defined as $F_\epsilon(\hat{\bm{\lambda}})=\mathbb{E}[\underset{{\bm{y} \in \mathcal{X}}}{\textbf{max}}  (\hat{\bm{\lambda}} +\epsilon \bm{Z})^T\bm{y}-u(\bm{y})]$, and the corresponding optimal solution is $\bm{y}_\epsilon^*(\hat{\bm{\lambda}})=\mathbb{E}[\underset{{\bm{y} \in \mathcal{X}}}{\textbf{argmax}}  (\hat{\bm{\lambda}} +\epsilon \bm{Z})^T\bm{y}-u(\bm{y})]$. Thus, the perturbed loss function is defined as:

\begin{equation} \label{eq: loss1}
\resizebox{0.49\textwidth}{!}{$
\begin{aligned}
 & \mathcal{L}_\epsilon^{PFY}(\hat{\bm{\lambda}},\bar{\bm{y}} )= \mathbb{E}[\underset{\bm{y}\in \mathcal{X}}{\textbf{max}}  \{ (\hat{\bm{\lambda}} +\epsilon \bm{Z})^T \bm{y}-u(\bm{y})\}]- (\hat{\bm{\lambda}}^T\bar{\bm{y}}-u(\bar{\bm{y}}))\\
\end{aligned}
$}
\end{equation}

\subsection{Convex and Smooth Perturbed Loss Function}
We will show in the following Propositions that the perturbed loss function is both convex and smooth, which would facilitate efficient training in the learning pipeline. Our work is motivated by recent work~\cite{BBT20}. Due to the space limit, we provide detailed proofs to all propositions in the Appendix\cite{YAX24}\footnote{The Appendix can be found at: https://arxiv.org/pdf/2406.17085}.
% The key to show the differentiability of the perturbed loss function lies in showing the $F_\epsilon(\hat{\bm{\lambda}})$ is differentiable.  
% First, Proposition 1 proves that $F_\epsilon(\hat{\bm{\lambda}})$ is differentiable and Lipschitz continuous.
% Due to the space limit, we provide detailed proofs in Appendix. Readers can refer to the Appendix for more details. Our proofs are motivated by recent work \cite{BMN20}.

% First, Proposition 1 proves that $F_\epsilon(\hat{\bm{\lambda}})$ is differentiable and Lipschitz continuous.

\begin{proposition}\emph{The differentiablity of perturbed function.}
As noise $\bm{Z}$ is from Gaussian distribution, it has the density $\rho(\bm{Z})\propto \text{exp}(-\psi(\bm{Z}))$. For $\mathcal{R}_{\mathcal{X}}=\underset{{y\in \mathcal{X}}}{\text{max}}||\bm{y}_{\epsilon}^*(\hat{\bm{\lambda}})||$, we have
 \begin{itemize}
 \item $F_{\epsilon}(\hat{\bm{\lambda}}) $ is differentiable, and 
$\nabla_{\hat{\bm{\lambda}}} F_{\epsilon}(\hat{\bm{\lambda}})=\bm{y}_{\epsilon}^*(\hat{\bm{\lambda}}) =\mathbb{E}[ y^*(\hat{\bm{\lambda}}+\epsilon \bm{Z})]=\mathbb{E}[F(\hat{\bm{\lambda}}+\epsilon \bm{Z})\nabla_Z \psi(\bm{Z})/\epsilon]$. \\
 \item $F_{\epsilon}(\hat{\bm{\lambda}})$ is Lipschitz continuous, and  $||F_{\epsilon}(\hat{\bm{\lambda}})-F_{\epsilon}(\tilde{\bm{\lambda}})||
 \leq \mathcal{R}_{\mathcal{X}}|| \hat{\bm{\lambda}}-\tilde{\bm{\lambda}}||.$
\end{itemize}  
\end{proposition}

\noindent where $\propto$ denotes ``proportional to.'' Proposition 1 proves that $F_\epsilon(\hat{\bm{\lambda}})$ is differentiable and Lipschitz continuous, we demonstrate that by adding appropriate Gaussian noise $\bm{Z}$, the perturbed function becomes differentiable with respect to $\hat{\bm{\lambda}}$. {The non-differentiability issue in the non-perturbed problem \eqref{obj1} arises from the linear energy storage model, where linear programming often involves a piecewise linear structure.
 The solution from the linear energy storage model typically lies on the vertices of a polytope. Consequently, small changes in the input price predictions can cause abrupt shifts in the solutions, leading to non-differentiability of the corresponding loss function. Adding perturbation smooths these abrupt transitions and ensures the perturbed optimization function $F_\epsilon(\hat{\bm{\lambda}})$ differentiable. } The equations from the first statement of Proposition 1 will be used in subsequent proofs. The Lipschitz continuity of $F_\epsilon(\hat{\bm{\lambda}})$ helps establish the Lipschitz continuity of ${\nabla \mathcal{L}_\epsilon^{PFY}(\hat{\bm{\lambda}},\bar{\bm{y}})}$ in Proposition 4. { The Lipschitz continuity ensures that the function does not change abruptly and the function is continuous.} Given the perturbation, we aim to understand the impact of $\epsilon$ and how the optimal solution $\bm{y}_{\epsilon}^*(\hat{\bm{\lambda}})$ behaves as $\epsilon$ varies. Consequently, we introduce Proposition 2:

\begin{proposition}\emph{Impact of scaling factor $\epsilon$.}
 \begin{itemize}
 \item $\bm{y}_{\epsilon}^*(\hat{\bm{\lambda}})$ is in the interior of $\mathcal{X}$. 
 \item For $\hat{\bm{\lambda}}$ such that $\bm{y}^*({\hat{\bm{\lambda}}})$ is a unique maximizer: when $\epsilon \rightarrow \infty$, $\bm{y}_{\epsilon}^*(\hat{\bm{\lambda}})\rightarrow \bm{y}_{1}^*(\frac{\hat{\bm{\lambda}}}{\epsilon})=\epsilon \mathbb{E}[ \underset{{\bm{y}\in \mathcal{X}}}{\textbf{argmax}} \{\bm{Z}^T\bm{y}\}]$; when  $\epsilon \rightarrow 0$, $\bm{y}_{\epsilon}^*(\hat{\bm{\lambda}})\rightarrow \bm{y}^*({\hat{\bm{\lambda}}}).$
\end{itemize}  
\end{proposition}

Proposition 2 indicates that  the perturbed  solution $\bm{y}_\epsilon^*({\hat{\bm{\lambda}}})$ remains feasible. If $\bm{y}^*({\hat{\bm{\lambda}}})$ is a unique maximizer, as $\epsilon \rightarrow 0$, $\bm{y}_{\epsilon}^*(\hat{\bm{\lambda}})$ converges to the original optimal solution $\bm{y}^*({\hat{\bm{\lambda}}})$. Conversely, as $\epsilon \rightarrow \infty$, $\bm{y}_{\epsilon}^*(\hat{\bm{\lambda}})$ deviates from the original solution $\bm{y}^*({\hat{\bm{\lambda}}})$ and approaches $\bm{y}_{1}^*(\frac{\hat{\bm{\lambda}}}{\epsilon})$. Since $\bm{Z}$ is drawn from a zero-mean Gaussian distribution, $\bm{y}_{1}^*(\frac{\hat{\bm{\lambda}}}{\epsilon})$ converges to a zero vector. In the context of our energy storage model, as $\epsilon$ tends to infinity, it implies that the energy storage system remains inactive and does not take any action.

This proposition implies that we should carefully select the scaling factor $\epsilon$. A very small $\epsilon$ may not ensure sufficient smoothness of the loss function, while a very large $\epsilon$ may cause the solution to deviate significantly from the original optimal solution. In our experimental setup, $\epsilon$ is set between 1 and 10.

%\noindent \textit{Remark.} It's important to note that although as $\epsilon$ tends to 0, $\bm{y}_{\epsilon}^*(\hat{\bm{\lambda}})$ converges to the original solution, setting $\epsilon$ close to zero is not feasible as the neural network requires a smooth loss landscape for backpropagation. In our experimental setup, $\epsilon$ is set between 1 and 10.

\begin{proposition}\emph{The property of perturbed loss function $\mathcal{L}_\epsilon^{PFY}(\hat{\bm{\lambda}},\bar{\bm{y}} )$.}
 \begin{itemize}
 \item $\mathcal{L}_\epsilon^{PFY}(\hat{\bm{\lambda}},\bar{\bm{y}} )$ is a  convex function of  ${\hat{\bm{\lambda}}}$.
 \item $ 0\leq \mathcal{L}_\epsilon^{PFY}(\hat{\bm{\lambda}},\bar{\bm{y}} ) -\mathcal{L}^{FY}(\hat{\bm{\lambda}},\bar{\bm{y}})  \leq TP \epsilon$.
\end{itemize}  
\end{proposition}

\noindent where $T$ represents the length of the prediction horizon, and $P$ represents the power rating. % This proposition shows that the perturbed loss function is convex and bounded. 
The left inequality in the second statement implies that the perturbed loss function serves as an upper bound for the original loss function, $\mathcal{L}^{FY}(\hat{\bm{\lambda}},\bar{\bm{y}}) \leq \mathcal{L}_\epsilon^{PFY}(\hat{\bm{\lambda}},\bar{\bm{y}})$. The right inequality indicates that the distance between these two loss functions is bounded by the scaling factor $\epsilon$. As $\epsilon \rightarrow 0$, these two inequalities together imply the convergence $\mathcal{L}_\epsilon^{PFY}(\hat{\bm{\lambda}},\bar{\bm{y}}) \rightarrow \mathcal{L}^{FY}(\hat{\bm{\lambda}},\bar{\bm{y}})$.

\begin{proposition}\emph{Gradient of perturbed loss function.}

 \begin{itemize}
 \item The gradient of the perturbed loss function is
  \begin{equation} \label{eq: prop4}
\begin{aligned}
  & {\nabla \mathcal{L}_\epsilon^{PFY}(\hat{\bm{\lambda}},\bar{\bm{y}})} =\bm{y}_\epsilon^*(\hat{\bm{\lambda}})-\bar{\bm{y}}. 
\end{aligned}
\end{equation}
 \item ${\nabla \mathcal{L}_\epsilon^{PFY}(\hat{\bm{\lambda}},\bar{\bm{y}})} $ is Lipschitz continuous.
\end{itemize}  
\end{proposition}

%
%Because we have  the total gradient  $\frac{\partial  \mathcal{L}_\epsilon(\hat{\bm{\lambda}},\bar{\bm{y}} )}{\partial \bm{w}}=  \frac{\partial \mathcal{L}_\epsilon(\hat{\bm{\lambda}},\bar{\bm{y}} )}{\partial \hat{\bm{\lambda}}}\frac{\partial \hat{\bm{\lambda}}}{\partial \bm{w}}$

Equation $\eqref{eq: prop4}$ shows that the perturbed loss function has a convenient gradient, which can be computed by solving for the optimal solution and subtracting $\bar{\bm{y}}$. Lipschitz continuity ensures that the gradient changes smoothly, which is beneficial for training the learning problem. Hence, we have shown that the perturbed loss function is both convex and smooth.

\noindent \textit{Convergence Analysis of the Optimization Layer:} As shown in later equation \eqref{sub:total}, the total gradient with respect to the network weights $\bm{w}$ relies on the chain rule. However, conducting a comprehensive convergence analysis is beyond the scope of this paper as it requires several assumptions regarding neural networks to complete the proof. Without loss of generality, we focus on the perturbed optimization layer and consider $\mathcal{L}_\epsilon^{PFY}(\hat{\bm{\lambda}},\bar{\bm{y}})$ as a function of $\hat{\bm{\lambda}}$. Given that $\mathcal{L}_\epsilon^{PFY}(\hat{\bm{\lambda}},\bar{\bm{y}} )$ is a  convex function of  ${\hat{\bm{\lambda}}}$ in Proposition 3, and ${\nabla \mathcal{L}_\epsilon^{PFY}(\hat{\bm{\lambda}},\bar{\bm{y}})} $ is Lipschitz continuous in Proposition 4, we provide a  convergence analysis for the optimization layer in Appendix.

%\noindent \textit{Convergence Analysis of Optimization Layer.} As we will show in equation \eqref{sub:total}, the total gradient with respect to the network weights $\bm{w}$ is based on the chain rule.  Then the whole convergence analysis is beyond the scope of this paper as we need several assumptions for neural networks to complete the proof. Without the loss of generality, we view $\mathcal{L}_\epsilon^{PFY}(\hat{\bm{\lambda}},\bar{\bm{y}})$ as a function of $\hat{\bm{\lambda}}$. As we show $\mathcal{L}_\epsilon^{PFY}(\hat{\bm{\lambda}},\bar{\bm{y}} )$ is a  convex function of  ${\hat{\bm{\lambda}}}$ and ${\nabla \mathcal{L}_\epsilon^{PFY}(\hat{\bm{\lambda}},\bar{\bm{y}})} $ is Lipschitz continuous in Proposition 4, the proof of the convergence of optimization layer is straightforward. We provide proof of convergence analysis for the optimization layer in the Appendix.

% $\mathcal{L}_\epsilon^{PFY}(\hat{\bm{\lambda}},\bar{\bm{y}} )=\mathcal{L}_\epsilon^{PFY}(g_w(\bm{x}),\bar{\bm{y}} )= \sum_{i=1}^n \mathcal{L}(g_w(\bm{x}^i),\bar{\bm{y}}^i)= \sum_{i=1}^n \mathcal{L}(g_w(\bm{x}^i),\bar{\bm{y}}^i)$

% Then the gradient of $\mathcal{L}_\epsilon^{PFY}(\hat{\bm{\lambda}},\bar{\bm{y}} )$ is

% $\nabla_w \mathcal{L}_\epsilon^{PFY}(g_w(\bm{x}^i),\bar{\bm{y}} )=\sum_{i=1}^nJ_w g_w(x^i) (\bm{y}_{\epsilon}^*(g_w(x^i))-\bar{y}^i) $

%Reference \cite{BBT20} shows that adding the perturbation makes the objective function differentiable. 

\subsection{Algorithmic Implementation}
Since the exact computation of the expectation in \eqref{eq: loss1} is generally intractable, we approximate it using Monte Carlo sampling in the solution algorithm implementation. {Monte Carlo sampling approximates expectations by averaging over $K$ number of random samples. }The perturbed loss function is approximated as:
\begin{equation} \label{eq: perturbsedloss}
\begin{aligned}
 & \mathcal{L}_\epsilon^{PFY}(\hat{\bm{\lambda}},\bar{\bm{y}} ) \\
  & \approx \frac{1}{K}\sum_{m=1}^K[\underset{\bm{y} \in \mathcal{X}}{\textbf{max}}  \{ (\hat{\bm{\lambda}} +\epsilon \bm{Z}^{(m)})^T \bm{y}-u({\bm{y}})\}]- (\hat{\bm{\lambda}}^T\bar{\bm{y}}-u(\bar{\bm{y}}))
 %& = \frac{1}{K}\sum_{m=1}^K[ (\hat{\bm{\lambda}}+\epsilon \bm{Z}^{(m)})^T \bm{y}_\epsilon^*(\hat{\bm{\lambda}})-u(\bm{y}_\epsilon^*(\hat{\bm{\lambda}}))]- (\hat{\bm{\lambda}}^T\bar{\bm{y}}-u(\bar{\bm{y}}))\\
\end{aligned}
\end{equation}
%Also note that we do not need to perturb the $\hat{\lambda}_t$ for the second term because that term is continious. 
\noindent where  $K$ is the number of Monte-Carlo samples. The gradient ${\nabla \mathcal{L}_\epsilon^{PFY}(\hat{\bm{\lambda}},\bar{\bm{y}})}$ is computed as follows:

%\cite{BMN20} shows that the perturbed loss function is differentiable.\\

%\noindent 
%and it is so-called Fenchel-Young Loss \cite{BMN20}. We can define the Fenchal conjugate $\Omega^*(\theta)= \mathbb{E}[\textbf{max}_{p_t, b_t}  \{\sum_{t=1}^T (\hat{\lambda}_{t} +\epsilon Z) y_t-C p_t\}]- ({\sum_{t=1}^T \hat{\lambda}_{t} \bar{y}_t-C \bar{p}_t)$. The original function is $\Omega(\bar{p},\bar{b})$. Then $ \mathcal{L}_\epsilon^{FY}(\hat{\lambda},\bar{p},\bar{b} )+\Omega(\bar{p},\bar{b})$ is positive by Fenchel's inequality, convex and the minimum 0 is reached at $\hat{\lambda}$ if and only if $f_\epsilon(\hat{\lambda})=\bar{p},\bar{b}$.
 %The Fenchel-Young loss has some very nice properties, we can easily get the 
\begin{equation}
\begin{aligned}
 & {\nabla \mathcal{L}_\epsilon^{PFY}(\hat{\bm{\lambda}},\bar{\bm{y}})} =\bm{y}_\epsilon^*(\hat{\bm{\lambda}}) -\bar{\bm{y}} \\
&  = \mathbb{E}[\underset{\bm{y} \in \mathcal{X}}{\textbf{argmax}}
 (\hat{\bm{\lambda}} +\epsilon \bm{Z})^T\bm{y}-u(\bm{y})] -\bar{\bm{y}}\\
 &  \approx \frac{1}{K}\sum_{m=1}^K[\underset{{\bm{y}\in \mathcal{X}}}{\textbf{argmax}}
 (\hat{\bm{\lambda}} +\epsilon \bm{Z}^{(m)})^T\bm{y}-u(\bm{y})] -\bar{\bm{y}}
\end{aligned}
\end{equation}

Given that the unknown predictor is designed to forecast real-time prices, we can guide the inferred reward to align with the price information by leveraging it as prior knowledge. To achieve this, we introduce a predictor
regularizer, which minimizes the mean squared error (MSE) between the reward and the price. The price data may consist of either previous days' real-time prices or day-ahead prices. Consequently, the total loss function is constructed as the weighted sum of the perturbed loss $\mathcal{L}^{PFY}_{\epsilon}(\hat{\bm{\lambda}},\bar{\bm{y}})$ and the MSE loss. This hybrid loss function is defined as:
\begin{equation}\label{eq:hybrid}
\mathcal{L}_\epsilon(\hat{\bm{\lambda}},\bar{\bm{y}}) = \mathcal{L}^{PFY}_{\epsilon}(\hat{\bm{\lambda}},\bar{\bm{y}}) + \beta \sum_{t=1}^T(\hat{\lambda}_t- \xi_t)^2
\end{equation}
The parameter $\beta$ represents the weighted factor of the MSE loss, and $\xi_t$ corresponds to the RTP or DAP from the previous day, serving as one of the features in $x_t$. The total subgradient is then computed as:
\begin{equation} 
\begin{aligned}
 & {\nabla \mathcal{L}_\epsilon(\hat{\bm{\lambda}},\bar{\bm{y}})}
  =\bm{y}_\epsilon^*(\hat{\bm{\lambda}}) -\bar{\bm{y}} +2\beta (\hat{\bm{\lambda}}- \bm{\xi})\\
  &\approx \frac{1}{K}\sum_{m=1}^K[\underset{{\bm{y}\in \mathcal{X}}}{\textbf{argmax}}
 (\hat{\bm{\lambda}} +\epsilon \bm{Z}^{(m)})^T\bm{y}-u(\bm{y})] -\bar{\bm{y}}+\\&2\beta (\hat{\bm{\lambda}}- \bm{\xi})\\
\end{aligned}
\end{equation}

For the backward propagation, we apply the chain rule to compute the total gradient with respect to the weights $\bm{w}$ of predictors,
\begin{equation}\label{sub:total}
 {\nabla_{\bm{w}} \mathcal{L}_\epsilon(\hat{\bm{\lambda}},\bar{\bm{y}})}=  {\nabla \mathcal{L}_\epsilon(\hat{\bm{\lambda}},\bar{\bm{y}})} \cdot \frac{\partial \hat{\bm{\lambda}}}{\partial \bm{w}}
\end{equation}
\noindent where the term $\frac{\partial \hat{\bm{\lambda}}}{\partial \bm{w}}$ represents the Jacobian of reward prediction with respect to the weights $\bm{w}$. The Jacobian can be computed using automatic differentiation \cite{PGS17} in a standard deep learning framework. The details of behavior prediction are outlined in Algorithm 1. The energy storage arbitrage follows a similar process as described in Algorithm 2. Due to space limitations, Algorithm 2 is moved to the Appendix \cite{YAX24}. Readers can refer to the Appendix \cite{YAX24} for more details.

\begin{algorithm} [!ht]
\caption{The Prediction of Strategic Energy Storage Behavior} 
\begin{algorithmic}[1]  \label{alg1}
\REQUIRE The time horizon $T$; The power rating $P$ of energy storage; The SoC efficiency $\eta$ and capacity $E$;   The maximum epochs $M_{\max}$; The cost coefficients in $u(p_t, b_t)$; The training data $\mathcal{D}=\{(\bm{x}^1,\bar{\bm{y}}^1), (\bm{x}^2, \bar{\bm{y}}^2),...,(\bm{x}^n, \bar{\bm{y}}^n)\}$; The testing data $\hat{\mathcal{D}}=\{\hat{\bm{x}}^1, \hat{\bm{x}}^2,..., \hat{\bm{x}}^n\}$. \\
\STATE 
\textbf{Initialization}: The weights $\bm{w}$ of predictor $g_w(\bm{x})$  are randomly initialized.\\
\textbf{Training stage:}
 \FOR{Epochs < $M_{\max}$ }
 \FOR{Each batch in $\mathcal{D}$ }
 \STATE{Predict the reward using the predictor $\hat{\bm{\lambda}}=g_w(\bm{x})$; }
  \STATE{Forward pass to compute the optimal decisions $p_t^*, b_t^*, t=1,2,...,T$ by equation (\ref{eq:arb});}
 \STATE{Compute the loss $\mathcal{L}_\epsilon(\hat{\bm{\lambda}},\bar{\bm{y}})$ by equation (\ref{eq:hybrid})};
  \STATE{Compute the total gradient by equation (\ref{sub:total})  and backward pass to update the weights $\bm{w}$ of predictor.};
 \\

\ENDFOR
\ENDFOR

%\STATE{\bm{Y}=\mathcal{H}^{\dagger}\bm{X}}

\textbf{Prediction stage:}

\STATE{Predict the reward using the predictor $\hat{\bm{\lambda}}=g_w(\hat{\bm{x}})$; }

 \STATE{Forward pass to compute the optimal decisions $\hat{p}_t, \hat{b}_t, t=1,2,...,T$ by equation (\ref{eq:arb});}
%\STATE{Compute the  by ($\ref{hruncertaintyindex}$)}

%\STATE{Compute the uncertainty index $U_{\textrm{index}}$ by  (\ref{Uncertaintyindex2})}
%\STATE{Compute the predictive mean and uncertainty index by  and , respectively}
\RETURN
The predicted energy storage behaviors  $(\hat{p}_t, \hat{b}_t)$ for $t=1,2,...T$. %The uncertainty index $U_{\textrm{index}}$.

\end{algorithmic} 
\end{algorithm}

\section{Experiments}
The experiments are implemented in Python 3, Pytorch \cite{PGM19} and PyEPO package \cite{TK22} on a desktop with
3.0 GHz Intel Core i9, NVIDIA 4080 16 GB Graphic card, and  32 GB RAM. The optimization layer is implemented using the Groubi optimization package~\cite{gurobi}. The pre-built predictor is implemented in Keras \cite{keras} and Tensorflow \cite{AAB16}.  {The computational time for reward prediction and optimization is less than 0.1 seconds.}

The parameters for the energy storage arbitrage model are set as follows: power rating $P=0.5 MW$,  storage efficiency $\eta= 0.9$, initial SoC $e_0=0.5 MWh$, storage capacity $E=2 MWh$, prediction horizon $T=24$. The number of Monte-Carlo samples is $K=1$. 

 {In this paper, we employ two neural network models: a Long Short-Term Memory (LSTM) network and a Multi-Layer Perceptron (MLP) network. The LSTM model architecture includes a single LSTM layer, followed by two fully connected layers with a hidden dimension of 64. A ReLU activation function is applied after each fully connected layer. For the MLP model, the architecture consists of three fully connected layers, each with a hidden dimension of 96, with ReLU as the activation function. The Adam optimizer is used for model training. The learning rate is set as 0.01. The input features, such as demand and price data, are structured as time-series historical data, enabling the model to effectively capture temporal dependencies. The neural network outputs a time series of reward predictions, while the optimization layer generates energy storage decisions that align with the reward prediction horizon. In our experiments, we used a 24-hour look-back window as input, and the model predicts the reward and corresponding storage decisions for the subsequent 24-hour period.}

\subsection{Self-Scheduling Energy Storage Arbitrage}\label{exp:arb}

We conduct experiments focusing on Behind-the-Meter energy storage arbitrage, where the energy storage system autonomously manages its scheduling based on future price predictions at hourly real-time prices~\cite{CZZ17}. The reward forecasting horizon spans 24 hours and is updated hourly.  Therefore, the corresponding energy storage schedule is dynamically adjusted on an hourly basis. {We collect data from the New York Independent System Operator (NYISO), using data from 2017-2020 for training and 2021 for testing.
% Historical data of RTP, DAP, and demand data are collected from the New York Independent System Operator (NYISO) within the New York City load zone in 2017-2021.  Four-year data are used for training and one-year data is used for testing.
The historical data for Real-Time Price (RTP), Day-Ahead Price (DAP), and load data are collected. The data resolution is one sample per hour. The dataset includes both input features and target decision values. The target decisions are obtained by solving \eqref{eq:arb} using ground truth RTP. The input features consist of three components: RTP, DAP, and load data. These input features are typical for training a price forecaster. To augment the training samples in the dataset, a rolling horizon window with a step size of one hour is employed. The input features of the training dataset are structured as a tensor of size 
$35017\times 24 \times 3$, where 35017 represents the number of samples, 24 corresponds to hourly intervals, and 3 denotes the number of feature types. The target decision data is represented as a 
$35017 \times 24$ matrix. The testing dataset follows the same format, comprising 8713 samples.}

The arbitrage profit is computed as
\begin{equation}
    \textstyle\sum_{t=1}^T \bar{\lambda}_{t} (p_t^*-b_t^*)-u(p^*_t,b_t^*)
\end{equation}
\noindent where $\bar{\lambda}_t$ denotes the real-time price at time $t$; $p_t^*$ and $b_t^*$ denote the scheduled discharge and charge at time $t$, respectively. 

\begin{figure}[!ht]
  \centering
   \subfigure[]{ \includegraphics[width = .8\linewidth]{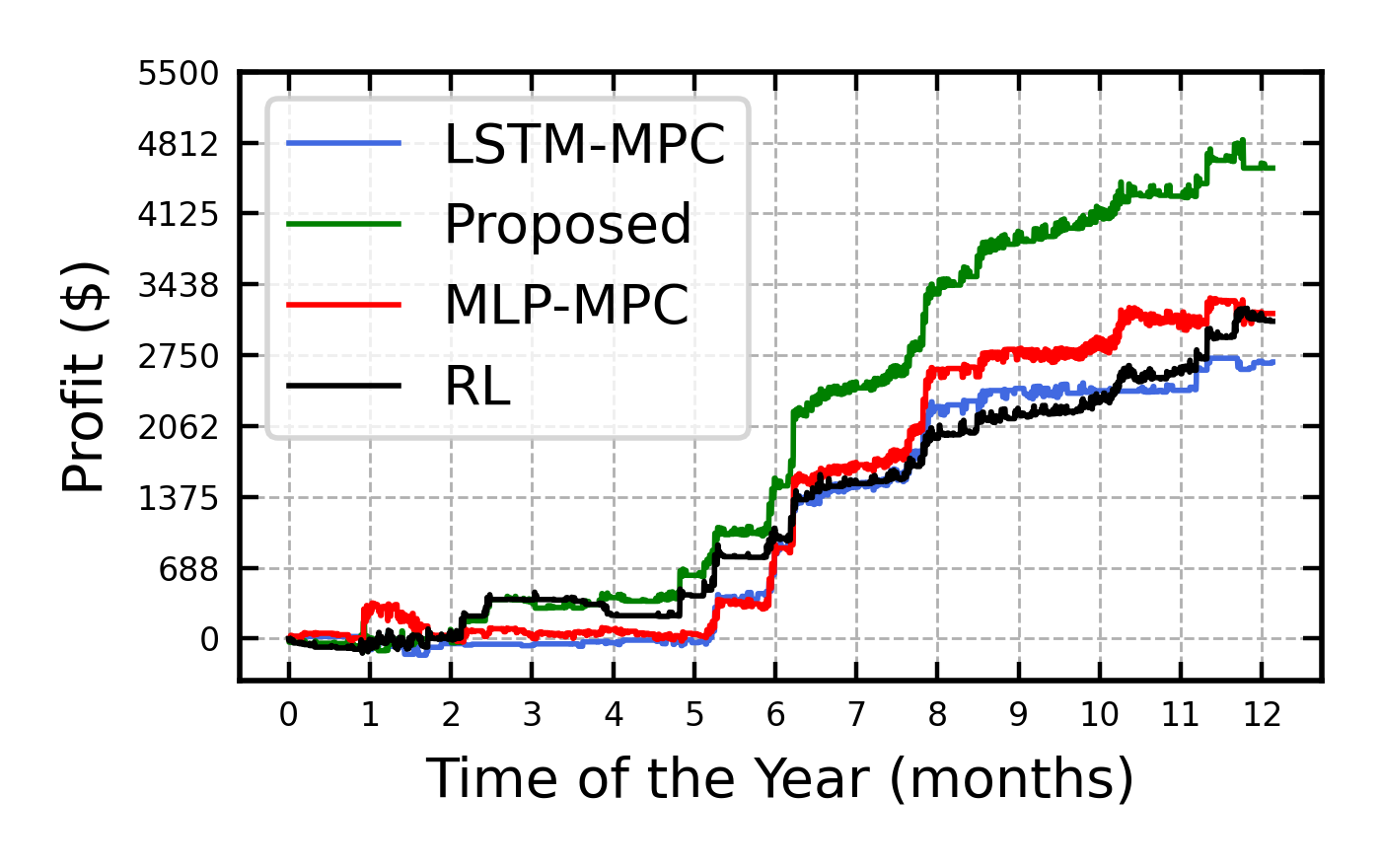}}
  \subfigure[]{ \includegraphics[width = .8\linewidth]{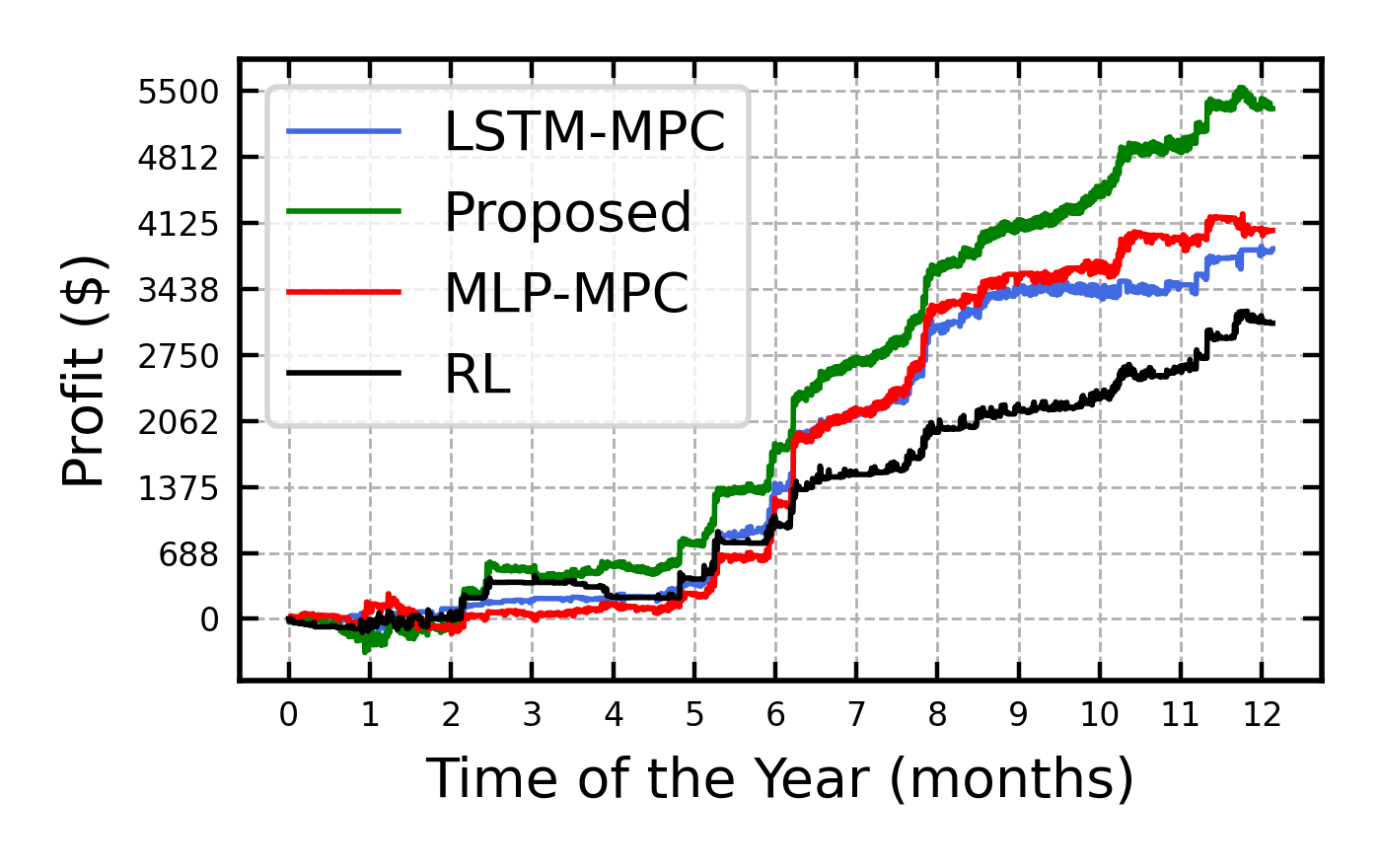} }
    	\caption{{Comparison of annual accumulative profits between the proposed approach and three benchmark methods. The energy storage model 
 is with (a) linear cost term, (b) linear and quadratic cost terms.}} \label{Fig: comp_arb1}
\end{figure}

{The goal of this experiment is to compare the proposed decision-focused approach with several data-driven benchmark approaches to demonstrate that training with our decision-focused loss leads to increased profit. We compare the cumulative profit over the entire year of 2021 and show the performance improvements achieved by our method.} We employ the MLP as our prediction layer. We compare the arbitrage performance of the proposed method with two MPC-based benchmark methods. The first benchmark employs an LSTM network for real-time price prediction, while the second employs an MLP network. For brevity, we denote these two methods as ``LSTM-MPC'' and ``MLP-MPC'', respectively. {We also include a benchmark method based on reinforcement learning (RL) \cite{WZ18}, which employs a Markov decision process (MDP) framework by discretizing the SoC. The original RL approach assumes perfect energy storage efficiency ($\eta = 1$). To ensure a fair comparison, we revised the RL approach by incorporating the SoC transition into the state transition of the RL framework and set the storage efficiency of all methods to 0.9.}
We test two types of energy storage models to evaluate the performance of our proposed method: one with a linear cost term and the other with both linear and quadratic cost terms. The linear cost term is $u(p_t,b_t)=C_1p_t$, where
the coefficient of discharge cost is $C_1= 10$. The linear and quadratic cost terms are $u(p_t,b_t)=C_1p_t+C_2p_t^2$, where
the coefficients of discharge cost are $C_1= 5$ and $C_2=5$.  

Note that the self-scheduling energy storage arbitrage setting in this paper is more challenging than the economic bidding or price response settings in \cite{BZX23}. Economic bidding assumes that energy storage can submit a bid, and the market clears it based on the price at time $t$.  Price response assumes that the price at the time step $t$ can be observed before making decisions at $t$. However, self-scheduling arbitrage in this paper requires making decisions at $t$ before the real-time price at $t$ is published. Consequently, this self-scheduling setting encounters more uncertainties and leads to lower profits than the other two settings.
 The accumulative arbitrage profits are plotted on a yearly basis.  
 
Fig. \ref{Fig: comp_arb1} shows the arbitrage performance of the proposed method and the three benchmark methods with the two considered cost terms. {The mean absolute error (MAE) of the MLP model is 10.83, while the MAE of the LSTM model is 9.31. Although the lower MAE of the LSTM  indicates better overall prediction accuracy compared to the MLP, the MLP model performs better in terms of profit. This is because the MLP predicts the price peak timing more accurately, a crucial factor for optimizing energy storage arbitrage.} The proposed method consistently outperforms the three benchmark methods.  
% For the energy storage model with linear cost,  the yearly accumulative profit of perfect forecasting is \$15728. The yearly profit obtained by the proposed method amounts to \$4658, while the profit of the MLP-MPC method and the LSTM-MPC method is \$2942 and \$2377, respectively. 
Our proposed method demonstrates a performance improvement of approximately $47\%$, $71\%$, and $50\%$ over the three benchmark methods with linear cost and 30\%, 38\%, and 72\% with linear and quadratic cost. 
%For the energy storage model with linear and quadratic cost.
%,  the yearly profit achieved by the proposed method is \$5162, while the MLP-MPC and LSTM-MPC methods have profits of \$4049 and \$3221, respectively. Our proposed method shows a profit improvement of approximately 27\% and 60\% over these two methods. 
% Because the benchmark methods do not leverage the structure of the energy storage model, the prediction and optimization are two independent problems. 
Our end-to-end pipeline, which incorporates the optimization model and focuses on decision error, contributes to its profit improvement. The results in the two sub-figures also demonstrate that our proposed method is robust across different energy storage cost models.

{We vary the energy storage efficiency and show the corresponding profits in Table \ref{table:arb}. The profits decrease as the efficiency decreases for all three methods. The proposed method achieves 5\%-20\% profit improvement over the two benchmark methods when $\eta=1$. In contrast, for $\eta=0.85$, the profits of the proposed method are more than double those of the two benchmark methods.
Unlike the two benchmark methods, which treat prediction and optimization as two separate stages, the proposed method incorporates the SoC evolution into the end-to-end pipeline.  The results in Table \ref{table:arb} demonstrate that the proposed method is more robust to varying storage efficiencies, highlighting its superiority.}

%exp linear cost: prop:4589, MLP: 3121, lstm 2583. 
%exp linear quadratic cost: prop: 5257, MLP: 4210, LSTM: 3791
 \begin{table}[]
 \centering
 \caption{The comparison of accumulative profits with different energy storage efficiency with linear cost.}  \label{table:arb}
\begin{tabular}{lllll}
 \toprule
$\eta$      & 1    & 0.95 & 0.9 & 0.85 \\
\hline 
Proposed & 9126 & 6381 &  4589   & 2990 \\
MLP-MPC      & 8388 & 5545 &   3121  & 1332 \\
LSTM-MPC     & 6910 & 4554 &   2583  & 1233 \\
\bottomrule
\end{tabular}
\end{table}

% \begin{figure}[!ht]
%   \centering
%   \includegraphics[width = 0.8\linewidth]{./Figures/exp5_cumsum_plot_profit.png} 
% %  {(a)}  \includegraphics[width = 1\linewidth]
%     %{./Figures/exp5_bar_plot_profit.png} 
%     %{(b)}
%     	\caption{The comparison of accumulative profits over one year between the proposed approach and two benchmark methods.  The energy storage model 
%  is with linear and quatratic cost term.  } \label{Fig: comp_arb2}
% \end{figure}

{In Fig.~\ref{Fig: comp_visual}, we present the predictions of the proposed method, LSTM, and MLP models, and the ground truth real-time price (RTP) over two consecutive days. We also show the corresponding energy storage decisions based on these predictions. As shown in Fig.~\ref{Fig: comp_visual}, while the reward predictions from all three methods exhibit similar trends, our method provides more accurate timing of price spikes, resulting in the best decision-making outcomes. The energy storage behavior from our method aligns closely with the optimal decisions. In contrast, the decisions generated by the LSTM and MLP models deviate from the optimal decisions and show larger errors. The profits for this segment are $\$363$, $\$220$, and $\$204$ for our method, MLP method, and LSTM method, respectively.}

\begin{figure}[!ht]
  \centering
{ \includegraphics[width = 1\linewidth]{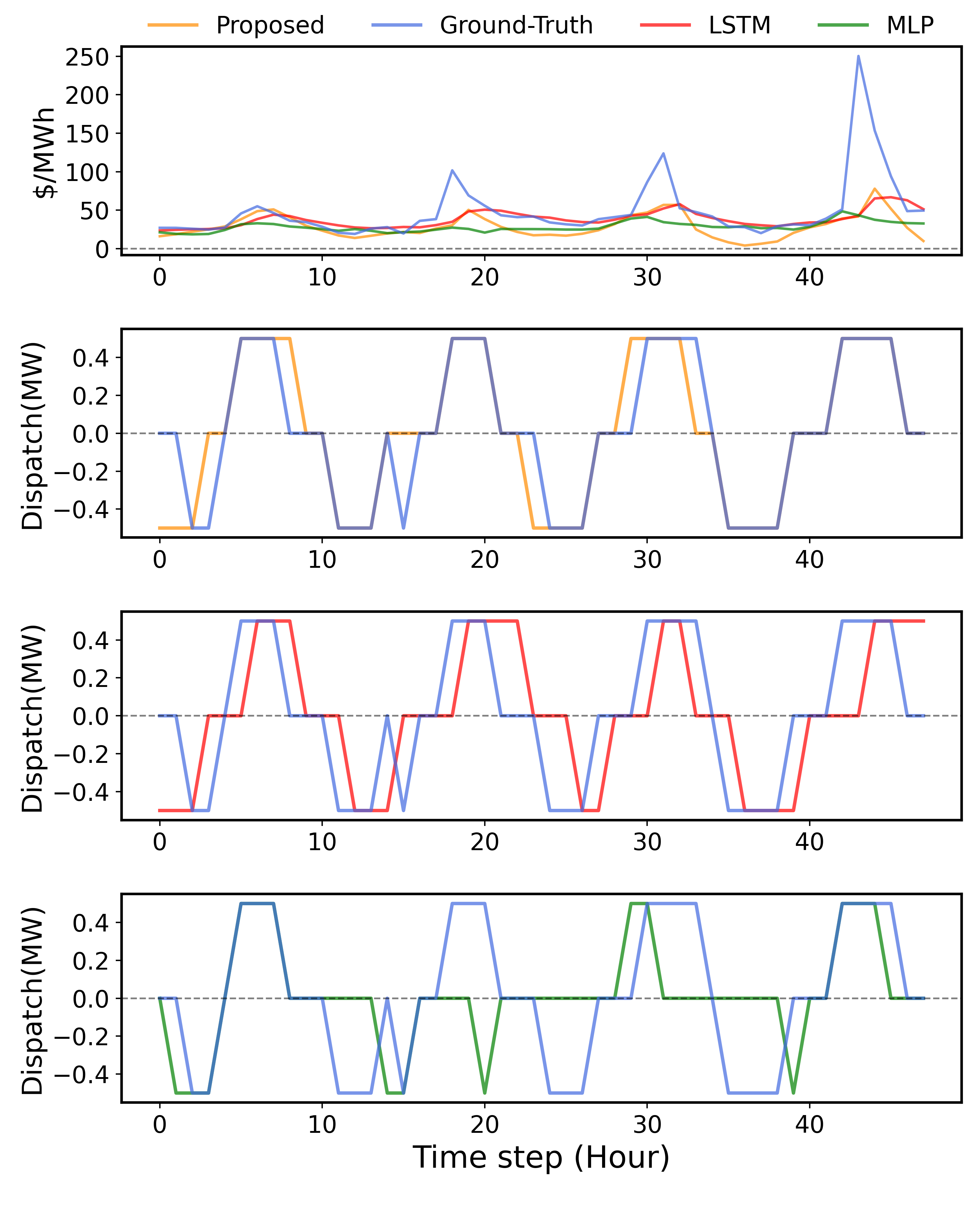} }
\caption{{Comparison of reward predictions and corresponding decisions between the proposed approach, two benchmark methods, and ground truth. The first subfigure shows the reward predictions and the ground truth real-time price (RTP). The subsequent subfigures show the corresponding decisions compared to the optimal decisions.}} \label{Fig: comp_visual}
\end{figure}

{To assess the adaptability of the proposed method to a new energy storage system, we employ an energy storage system with parameters $\eta = 0.95$, $E = 1$, $P = 0.5$, $C_1 = 10$, and $C_2 = 0$. A reward predictor is trained using the optimal decisions generated by this energy storage model. The trained predictor is then applied to a different model with parameters $\eta = 0.9$, $E = 2$, $P = 0.5$, $C_1 = 5$, and $C_2 = 5$. The arbitrage profit achieved through adaptation is \$4,783, compared to \$5,314 from the original model. We abbreviate the proposed method, trained on a storage system with different parameters, as ``Proposed-Adapt.'' The performance comparison in Fig.~\ref{Fig: comp_adapt} shows that the algorithm adapts effectively to the new model, with only a minor reduction in profit. This demonstrates the robustness of the trained predictor in handling variations in system parameters while maintaining competitive performance.}

\begin{figure}[!ht]
  \centering
{ 
   \includegraphics[width = 1\linewidth]{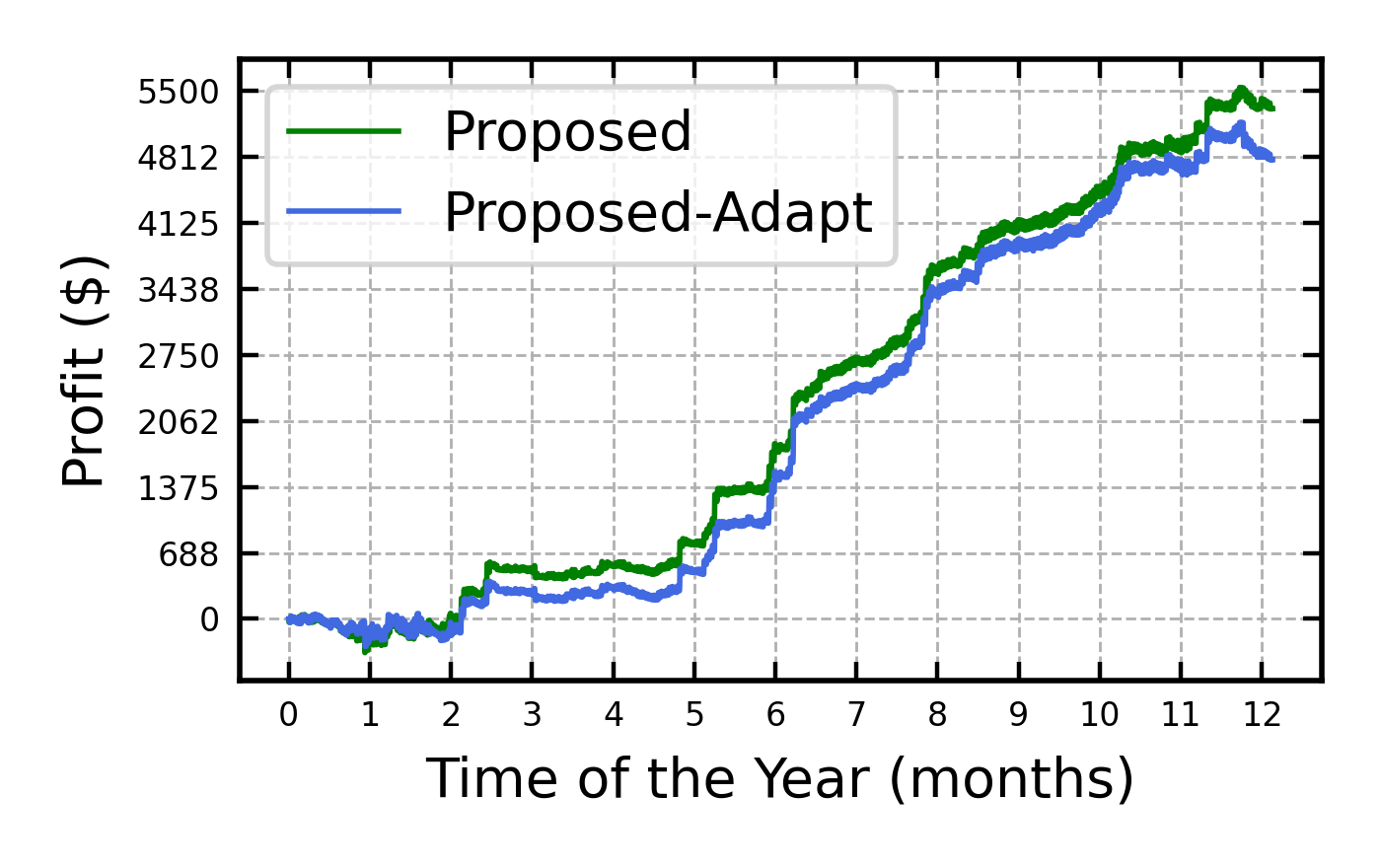} }
    	\caption{{Arbitrage performance based on training on the original energy storage system and on a storage system with different parameters.}}
\label{Fig: comp_adapt}
\end{figure}
\subsection{Energy Storage Behavior Prediction}
{The prediction of energy storage behavior involves predicting the charge and discharge patterns of an energy storage unit based on observed energy storage actions and input features.}
To evaluate the performance of our proposed algorithm, we compared it with two benchmark prediction methods. The first benchmark algorithm trains an LSTM network to learn a nonlinear mapping between input features and charge/discharge decisions. This mapping resembles the one in \eqref{eq:mapping_nn} but removes the optimization layer. It can be expressed as,
\begin{equation}\label{eq:mapping_nn2}
\bm{x} \mapsto \bm{y} = g_w(\bm{x})
\end{equation}

We denote the first algorithm as its neural network type ``LSTM'' for brevity. The second algorithm is a two-stage approach and assumes a pre-built predictor. After predicting real-time price, the price predictions are fed into the storage optimization problem defined in equation (\ref{eq:arb}). Due to its two-stage nature, we abbreviate this method as ``Two-Stage.'' 

{The goal of this experiment is to demonstrate that the proposed decision-focused framework improves the accuracy of predicting energy storage behavior. The LSTM method does not incorporate the energy storage model during training. The two-stage method separates prediction and energy storage optimization into separate parts, training the prediction model solely to minimize the price prediction error without considering the observed charge/discharge behaviors. Our proposed method integrates the energy storage model into the machine learning pipeline, minimizing decision loss with respect to observed charge/discharge actions. We evaluate prediction performance over a full year to demonstrate the effectiveness of our approach.}

%~\cite{bishop2006pattern}
\textit{Evaluation Metric.} 
We employ principles of pattern recognition to assess the efficacy of storage behavior prediction models. Our motivation stems from the unique characteristics of storage operations, which are inherently bi-directional—alternating between charging and discharging phases—and marked by the action's sparsity, where the output of the storage remains at zero for extended durations. Traditional evaluation metrics, including mean-square error and correlation coefficients, fall short of accurately capturing the prediction quality due to these distinct operational dynamics.

We assign the label 1 to the discharge decision, -1 to the charge decision, and 0 to the standby decision.  Predicted values are classified as 1 if above the threshold, -1 if below the negative threshold, and 0 if within the range between the negative threshold and the threshold. The threshold is set at 10\%-20\% of the power rating. In the synthetic experiment, we set the threshold to 0.05; in the real-world experiment, we set it to 0.2. We consider the time delays, and if the prediction and the ground truth differ by at most 2 hours, we still consider it correctly classified. The confusion matrix with this criterion is denoted as the ``event-based confusion matrix.'' The confusion matrix is defined in Table \ref{conf_matrix}, where TP is true positive; TN is true negative; FN is false negative; FP is false positive. 

\begin{table}[H]
  \caption{The confusion matrix of energy storage behavior.}
  \centering
  \label{conf_matrix}
  \begin{tabular}{cccc}
    \toprule
  Actual | Prediction & -1&0 &1\\
    \midrule
    -1 & TP& FN &FP\\
   0 & FP & TN&FP\\
    1 & FP & FN& TP\\
  \bottomrule
\end{tabular}
\end{table}

Due to the event-based criterion not considering the magnitude of the prediction, we also introduce a magnitude-based criterion: if the prediction error for charge/discharge falls within a certain percentage of magnitude of the actual value, it is classified as 1 or -1, respectively; otherwise, it is labeled as 0. In the synthetic experiment, we set the percentage at 20\%; in the real-world experiment, we set it at 40\%. The time delays are also accounted for in the magnitude-based criterion. We denote the confusion matrix with this criterion as the ``magnitude-based confusion matrix.'' 
The evaluation metrics are defined as:
\begin{subequations}
\begin{align}
    \text{Precision} &= \mathrm{\frac{TP}{TP+FP}}\\
    \text{Accuracy} &= \mathrm{\frac{TP+TN}{TP+TN+FP+FN}}\\
    \text{Recall} &= \mathrm{\frac{TP}{TP+FN}}\\
    \text{F1 score} &= \mathrm{\frac{2\times Precision \times Recall}{Precision+Recall}}
\end{align}
\end{subequations}

% \begin{table*}
%   \caption{Some Typical Commands}
%   \label{tab:commands}
%   \begin{tabular}{ccl}
%     \toprule
%     Command &A Number & Comments\\
%     \midrule
%     \texttt{{\char'134}author} & 100& Author \\
%     \texttt{{\char'134}table}& 300 & For tables\\
%     \texttt{{\char'134}table*}& 400& For wider tables\\
%     \bottomrule
%   \end{tabular}
% \end{table*}

\subsubsection{Storage Behavior Prediction with Synthetic Data}
{The historical data for RTP, DAP, and load data are collected from NYISO from 2019 to 2022. We use two years for training and one year for testing. The data resolution is one sample per hour. The dataset collected for this study comprises both input features and target decision values. Similar to the application of energy storage arbitrage, we select typical input features for a price forecaster, including three types: RTP, DAP, and load data. The target decisions are the historical energy storage charge/discharge actions. The input features are structured as a $730 \times 24 \times 3$ tensor, where 730 represents the number of samples, 24 corresponds to 24 hourly intervals, and 3 denotes the feature types. The target decision data is structured as a $730 \times 24$ matrix. The testing dataset follows the same format and contains 365 samples.} We employ an LSTM network as our prediction layer. The energy storage arbitrage model in \eqref{eq:arb} generates the energy storage behaviors. The governing reward, i.e., the synthetic price prediction, $\hat{\bm{\lambda}} \in \mathbb{R}^{T}$ is generated by:

\begin{equation}
    \hat{\lambda}^g_{i,t} = \alpha_{i,t}* \lambda^{DA}_{i,t}+(1-\alpha_{i,t})* \lambda^{RT}_{i,t}+ \zeta_{i,t}
\end{equation}

\noindent where $\lambda^{DA}_{i,t}$ and $\lambda^{RT}_{i,t}$ represent the DAP and RTP at time step $t$ of day $i$, respectively.  $\zeta_{i,t}$ is the Gaussian noise generated from $\mathcal{N}(0,1)$. $\alpha_{i,t}$ is the parameter that controls the trade-off similarity between DAP and RTP. Note that this reward generation method is motivated by the fact that the expectation of RAP converges to DAP~\cite{tang2016model}. We randomly select $\alpha_{i,t}$ from (0.5,1) in our experiment. The two-stage method constructs a predictor for RTP based on the same historical data and then incorporates the prediction into the arbitrage model.

\begin{figure}[!ht]
  \centering
   \includegraphics[width = 1\linewidth]{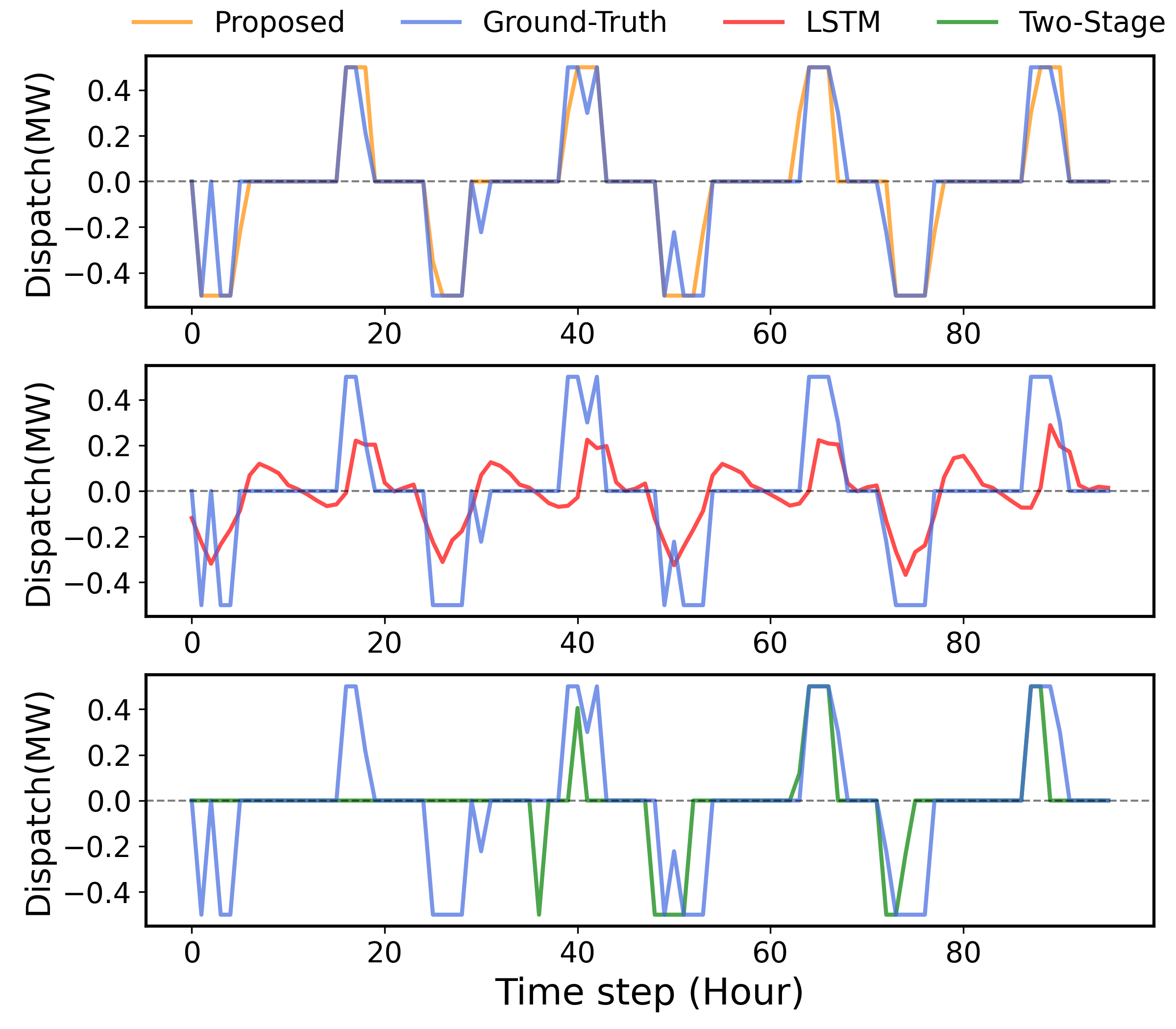} 
    	\caption{Comparison of ground-truth and predictions from the proposed approach and two benchmark methods: storage model with linear cost.} \label{Fig: comp1}
\end{figure}

We evaluate the performance of our proposed method using {energy storage behaviors} with two cost terms. The settings for these cost terms are consistent with those described in Section \ref{exp:arb}. The hyperparameter $\beta$ is set as $\beta=0.001$. Fig. \ref{Fig: comp1} and  Fig. \ref{Fig: comp2}  show the prediction performance of the proposed method and two benchmark methods.  The ground-truth behaviors and the predictions in four consecutive days are plotted. The proposed method accurately predicts storage behaviors regarding time steps and magnitudes. In contrast, the LSTM method can predict the timing of charge/discharge decisions but struggles to learn power rating information from historical data. Additionally, it mispredicts charge/discharge behaviors during standby periods when energy storage remains inactive. The two-stage method can capture the power rating but predicts a shorter duration.  Additionally,  it fails to predict energy storage activities on the first two days. {We show the price predictions from the two-stage method with the actual predictions in Fig.~\ref{Fig: twostage}. As shown in Fig.~\ref{Fig: twostage}, there are deviations between the two-stage predictions and the actual predictions. While the two-stage method's predictions follow a similar overall trend to the actual predictions, the two-stage model underestimates the price spread during the first two days. Because the predicted price spread from the two-stage method is lower than the degradation cost, the storage system remains idle.}

\begin{figure}[!ht]
  \centering
    \subfigure[]{ 
   \includegraphics[width = 1\linewidth]{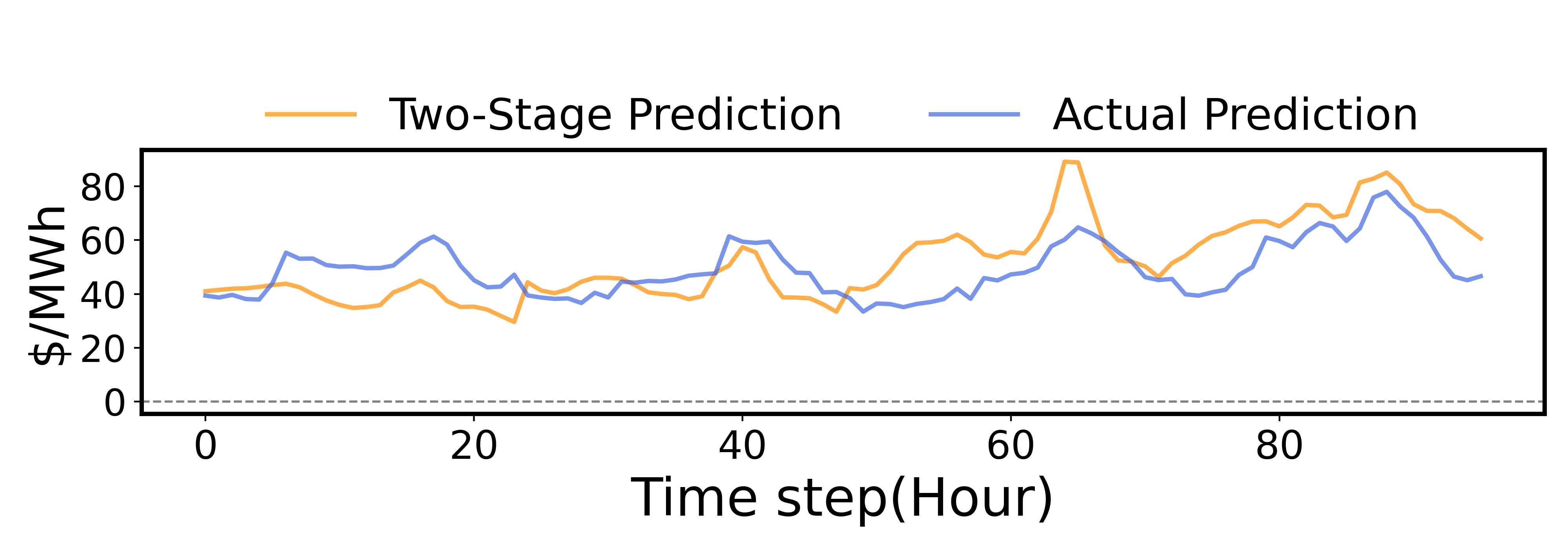} }
     \subfigure[]{ 
   \includegraphics[width = 1\linewidth]{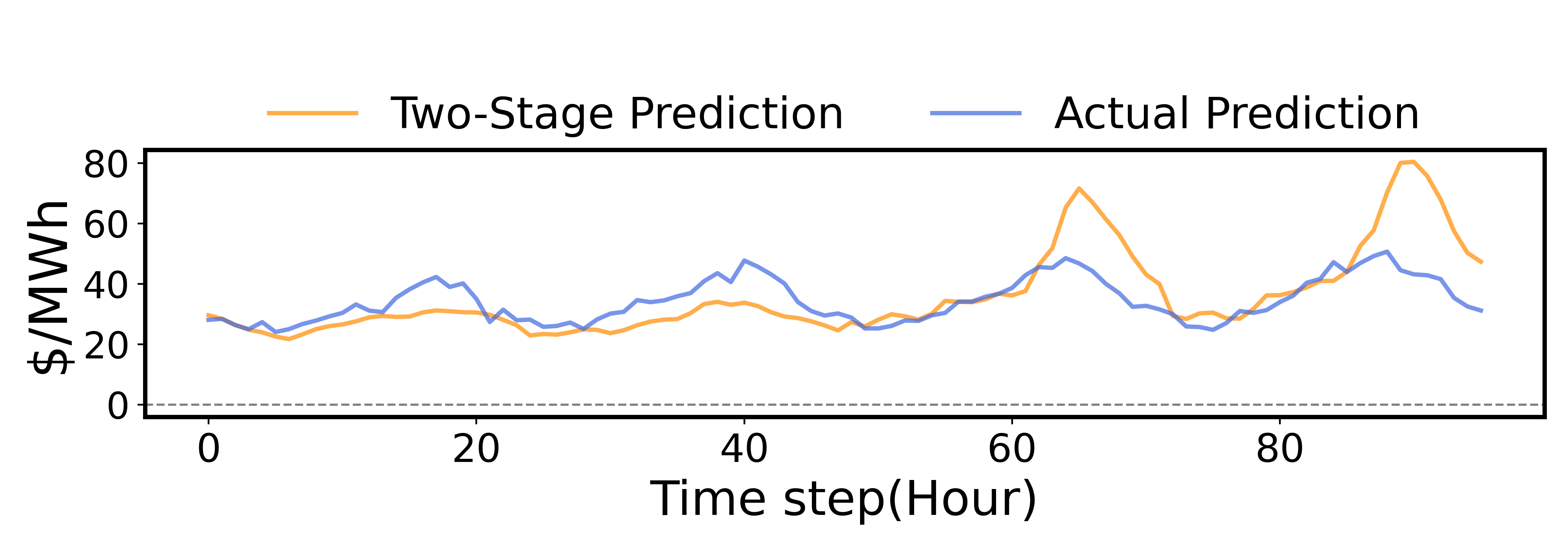}} 
    \caption{{Price prediction of the two-stage method and the actual prediction. (a) Price prediction of two-stage method in Fig.~\ref{Fig: comp1}; (b) Price prediction of two-stage method in Fig.~\ref{Fig: comp2}.}}
\label{Fig: twostage}
\end{figure}

\begin{figure}[!ht]
  \centering
   \includegraphics[width = 1\linewidth]{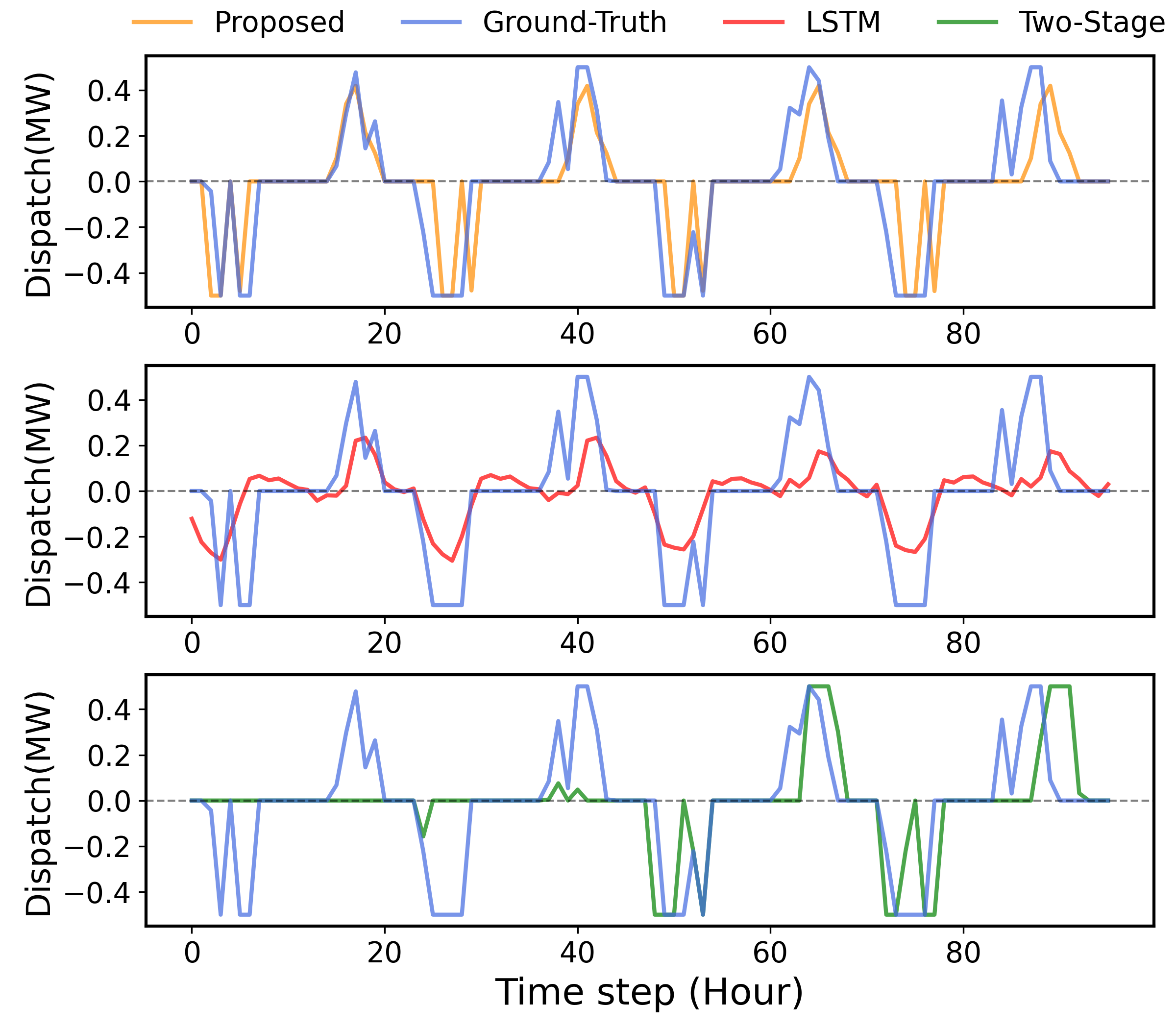} 
    	\caption{Comparison of ground-truth and predictions from the proposed approach and two benchmark methods: storage model with linear and quadratic cost.} \label{Fig: comp2}
\end{figure}

The total dispatch energy of ground-truth, the proposed method and two comparison methods is shown in Fig. \ref{Fig:barcomp1}(a) and (b). Our proposed method accurately predicts the charge and discharge energy while the two benchmark methods have higher errors.  This demonstrates the superiority of our approach.

Table \ref{exp1:conf_matrix1} and Table \ref{exp2:conf_matrix1} list the number of True Positives (TP), True Negatives (TN), False Positives (FP), and False Negatives (FN) along with two types of confusion matrices. The event-based confusion matrix shows that the proposed method achieves the highest precision and accuracy. While the proposed method has lower recall than LSTM, it achieves much higher precision. The F1 score, which balances precision and recall, indicates that our method outperforms others overall. The magnitude-based confusion matrix employs a more strict criterion. In both tables, the performance of the LSTM method degrades significantly under this criterion. In contrast, the F1 score of the proposed method experiences only a marginal decrease. This indicates that the proposed method is more adept at accurately predicting the magnitude. The F1 scores in both tables demonstrate that the proposed method outperforms the other two methods.

\begin{figure*}[!ht]
  \centering
\subfigure[]{ \includegraphics[width = 0.3\linewidth]{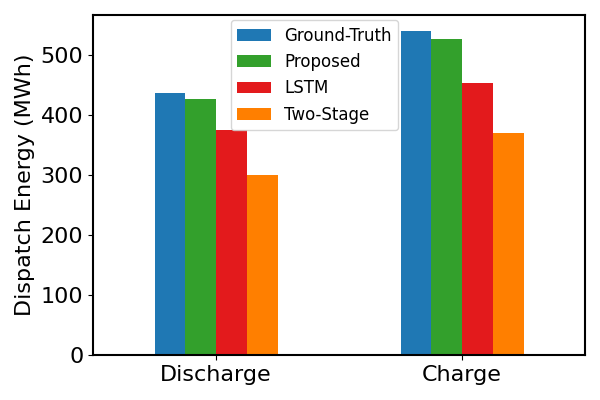}  }
  \subfigure[]{  \includegraphics[width = 0.3\linewidth]{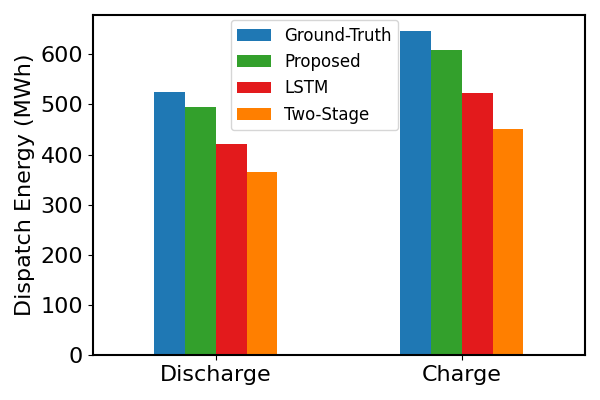} } 
\subfigure[]{  \includegraphics[width = 0.3\linewidth]{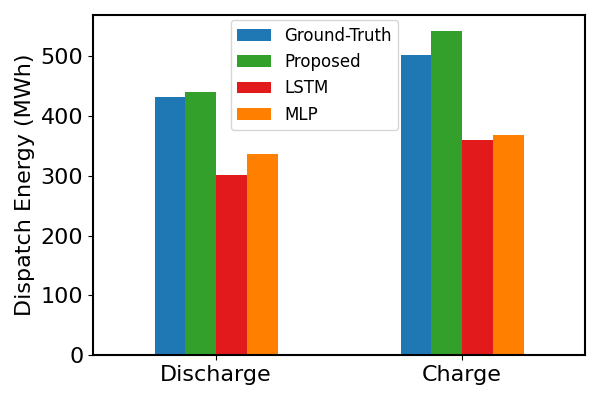}}      
    	\caption{Energy dispatch of ground-truth, the proposed approach and two benchmark methods: {energy storage behavior} with (a)  a linear cost term, (b)  linear and quadratic cost terms, and (c) real-world {energy storage behavior}.}\label{Fig:barcomp1}
\end{figure*}

\begin{table*} 
  \caption{Comparison of the proposed approach and two benchmark methods,  {storage behavior} with linear cost term.}
  \label{exp1:conf_matrix1}
    \resizebox{\textwidth}{!}{
  \begin{tabular}{c|cccccccc|cccccccc}
    \toprule
     & \multicolumn{8}{c}{Event-based confusion matrix}&\multicolumn{8}{c}{Magnitude-based confusion matrix}\\
     \midrule
  Metrics & TP&TN& FP &FN&Precision& Accuracy & Recall & F1 score& TP&TN& FP &FN&  Precision& Accuracy & Recall & F1 score\\
    \midrule
   Prop      & 1463 & 5878 & 672  & 723  & 68.52\% & 84.03\% & 66.93\% & 67.72\% & 1235 & 5808 & 742  & 951  & 62.47\% & 80.62\% & 56.50\% & 59.33\% \\
  LSTM      & 1987 & 4166 & 2418 & 165  & 45.11\% & 70.43\% & 92.33\% & 60.61\% & 132  & 4004 & 2580 & 2020 & 4.87\%  & 47.34\% & 6.13\%  & 5.43\%  \\
Two-Stage & 943  & 6015 & 548  & 1230 & 63.25\% & 79.65\% & 43.40\% & 51.47\% & 787  & 6034 & 529  & 1386 & 59.80\% & 78.08\% & 36.22\% & 45.11\%\\
  \bottomrule
\end{tabular}}
\end{table*}

%  \begin{table} 
%   \caption{The sensitivity analysis of parameter $\beta$ in hybrid loss function}
%   \label{exp1:sens1}
% \begin{tabular}{ccccc}
%     \toprule
%    $\beta$ & 0.0001& 0.001 & 0.01 & 0.1\\
%     \midrule
%    F1 score & 69.76\%& 70.07\% & 59.21\%&61.63\%\\
%   F1 score 2 &55.27\% & 62.02\%& 55.46\% &46.06\%\\
%   \bottomrule
% \end{tabular}
% \end{table}
% \begin{table} 
%   \caption{The comparison of proposed approach with two benchmark methods}
%   \label{conf_matrix}
%   \begin{tabular}{cccc}
%     \toprule
%   Metrics & Prop& LSTM & Two-Stage\\
%     \midrule
%     Precision & TP& FN &FP\\
%    Accuracy & FP & TN&FP\\
%    Recall & FP & FN& TP\\
%     F1 & FP & FN& TP\\
%   \bottomrule
% \end{tabular}
% \end{table}

% \begin{figure}[!ht]
%   \centering
%    \includegraphics[width = 1\linewidth]{./Figures/exp2_FY_LSTM_jul.png} 
%    \includegraphics[width = 1\linewidth]{./Figures/exp2_alg1_LSTM_jul.png} 
%   \includegraphics[width = 1\linewidth]{./Figures/exp2_alg2_LSTM_jul.png} 
  
%     	\caption{The comparison of the proposed approach against two benchmark methods, the {storage behavior} is with linear and quadratic cost} \label{Fig: comp2}
% \end{figure}

% \begin{figure}[!ht]
%   \centering
%    \includegraphics[width = 0.8\linewidth]{./Figures/exp2_bar_plot_energy.png}  
  
%     	\caption{The energy dispatch comparison of the proposed approach against two benchmark methods,  {storage behavior} with linear and quadratic cost term} \label{Fig:barcomp2}
% \end{figure}

%%%

\begin{table*} 
  \caption{Comparison of the proposed approach and two benchmark methods,  {storage behavior} with linear and quadratic cost term.}
  \label{exp2:conf_matrix1}
  \resizebox{\textwidth}{!}{
  \begin{tabular}{c|cccccccc|cccccccc}
    \toprule
   & \multicolumn{8}{c}{Event-based confusion matrix}&\multicolumn{8}{c}{Magnitude-based confusion matrix}\\
     \midrule
  Metrics & TP&TN& FP &FN&Precision& Accuracy & Recall & F1 score&TP&TN& FP &FN&Precision& Accuracy & Recall & F1 score\\
    \midrule
   Prop      & 2327 & 4796 & 1018 & 595  & 69.57\% & 81.54\% & 79.64\% & 74.26\% & 1252 & 4944 & 870  & 1670 & 59\%    & 70.92\% & 42.85\% & 49.64\% \\
LSTM      & 2577 & 3647 & 2242 & 270  & 53.48\% & 71.25\% & 90.52\% & 67.23\% & 228  & 3518 & 2371 & 2619 & 8.77\%  & 42.88\% & 8.01\%  & 8.37\%  \\
Two-Stage & 1430 & 5195 & 629  & 1482 & 69.45\% & 75.84\% & 49.11\% & 57.53\% & 906  & 5198 & 626  & 2006 & 59.14\% & 69.87\% & 31.11\% & 40.77\% \\
  \bottomrule
\end{tabular}}
\end{table*}

%%%

\begin{table*} 
  \caption{Comparison of the proposed approach and two benchmark methods on real-world data.}
  \label{exp3:conf_matrix}
  \resizebox{\textwidth}{!}{
  \begin{tabular}{c|cccccccc|cccccccc}
    \toprule
       & \multicolumn{8}{c}{Event-based confusion matrix}&\multicolumn{8}{c}{Magnitude-based confusion matrix}\\
     \midrule
  Metrics & TP&TN& FP &FN &  Precision& Accuracy & Recall & F1 score & TP&TN& FP &FN&  Precision& Accuracy & Recall & F1 score\\
    \midrule
Prop      & 792 & 4252 & 352 & 364 & 69.23\% & 87.57\% & 68.51\% & 68.87\% & 602 & 4272 & 332 & 554 & 64\%    & 84.62\% & 52.08\% & 57.61\% \\
LSTM      & 609 & 4297 & 293 & 561 & 67.52\% & 85.17\% & 52.05\% & 58.78\% & 345 & 4299 & 291 & 825 & 54.25\% & 80.62\% & 29.49\% & 38.21\% \\
MLP& 764 & 4024 & 580 & 392 & 56.85\% & 83.12\% & 66.09\% & 61.12\% & 367 & 4060 & 544 & 789 & 40.29\% & 76.86\% & 31.75\% & 35.51\%\\
  \bottomrule
\end{tabular}}
\end{table*}

\begin{table*} 
  \caption{{The performance of the proposed approach without assuming knowledge of the degradation cost, storage behavior with linear cost term and linear and quadratic cost term.}}
  \label{table:exp3}
    \resizebox{\textwidth}{!}{
  \begin{tabular}{c|cccccccc|cccccccc}
    \toprule
     & \multicolumn{8}{c}{Event-based confusion matrix}&\multicolumn{8}{c}{Magnitude-based confusion matrix}\\
     \midrule
  Metrics & TP&TN& FP &FN&Precision& Accuracy & Recall & F1 score& TP&TN& FP &FN&  Precision& Accuracy & Recall & F1 score\\
    \midrule
   Linear Cost &  1840& 5116& 1435& 345 & 56.18\%&  79.62\%&  84.21\%&  67.40\% & 1608& 5137&  1414& 577& 53.21\%&  77.21\%&  73.59\%& 61.76\%\\
  Linear and Quadratic     & 2197 & 4751 & 1082 & 706 & 67.00\% & 79.53\% & 75.68\% & 71.08\% &1623 & 4639& 1194& 1280 & 57.61\%& 71.68\% &55.91\%& 56.75\% \\
  \bottomrule
\end{tabular}}
\end{table*}

% \begin{figure}[!ht]
%   \centering
%    \includegraphics[width = 1\linewidth]{./Figures/exp3_FY_MLP_sept30.png} 
%    \includegraphics[width = 1\linewidth]{./Figures/exp3_alg1_lstm_Sept.png} 
%   \includegraphics[width = 1\linewidth]{./Figures/exp3_alg1_mlp_Sept30.png} 
  
%     	\caption{The comparison of proposed approach against two benchmark methods} \label{Fig: comp3}
% \end{figure}

\begin{figure}[!ht]
  \centering
   \includegraphics[width = 1\linewidth]{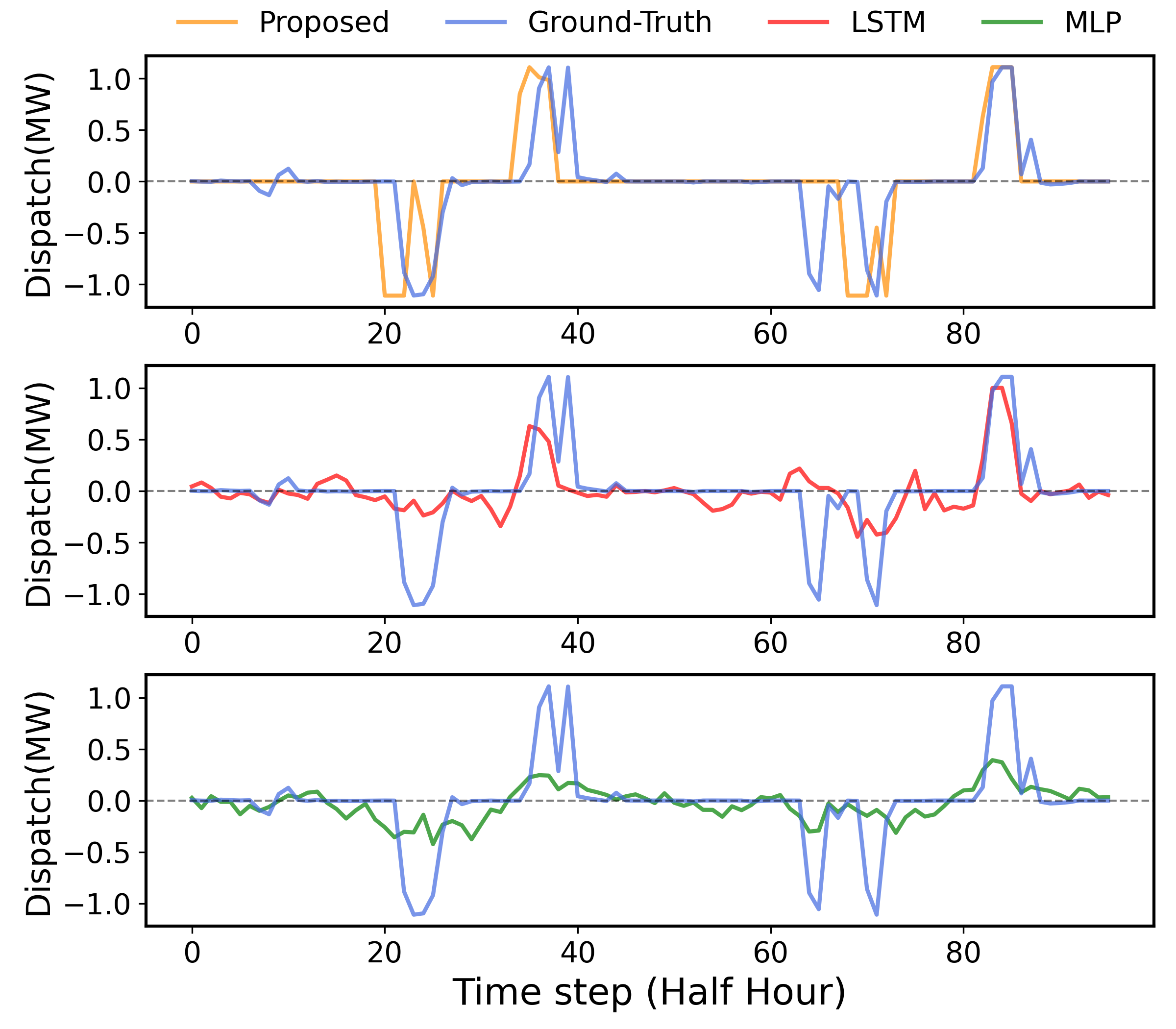} 
  
    	\caption{Comparison of ground-truth and predictions from the proposed approach and two benchmark methods: real-world data.} \label{Fig: comp3}
\end{figure}

% \begin{figure}[!ht]
%   \centering
%    \includegraphics[width = 0.8\linewidth]{./Figures/exp3_bar_plot_energy.png}  
%     	\caption{The energy dispatch comparison of the proposed approach against two benchmark methods,  real-world dataset} \label{Fig:barcomp3}
% \end{figure}

{Note that the regulator can obtain degradation cost information through resource integration registration. Our method does not rely on prior knowledge of linear cost coefficients, and the quadratic cost has only a minor influence on prediction accuracy. The algorithm can interpret both linear and quadratic cost coefficients as part of the reward prediction. However, because the reward prediction is linearly multiplied with charge/discharge decisions in the objective function, the quadratic cost cannot be fully captured, resulting in a slight impact on prediction accuracy.  In such cases, the algorithm can leverage data-driven storage model identification techniques~\cite{BZZ23} to estimate the energy storage degradation quadratic cost term based on inferred pricing dynamics from our approach.}

{In Fig.~\ref{Fig: comp1} and Fig.~\ref{Fig: comp2},  $C_1=10$ for linear case and $C_1=5$ and $C_2=5$ for linear and quadratic case. In Fig.~\ref{Fig: compv}, we do not assume that the degradation cost is known in our optimization model. We set $C_1=0$ in our model for the linear cost case. For the linear and quadratic cost case, we set $C_1=0$ and $C_2=0$. The results, shown in the Fig.~\ref{Fig: compv} and Table~\ref{table:exp3}, indicate that the algorithm achieves similar performance to the case where the ground truth cost terms are assumed. The algorithm accurately predicts the energy storage model under a linear cost. For energy storage with linear and quadratic costs, our method correctly predicts the timing of charge and discharge, with only minor errors in discharge predictions. This demonstrates that our algorithm is not sensitive to prior knowledge of the cost function.}

\begin{figure}[!ht]
  \centering
  \subfigure[]{  \includegraphics[width = 1\linewidth]{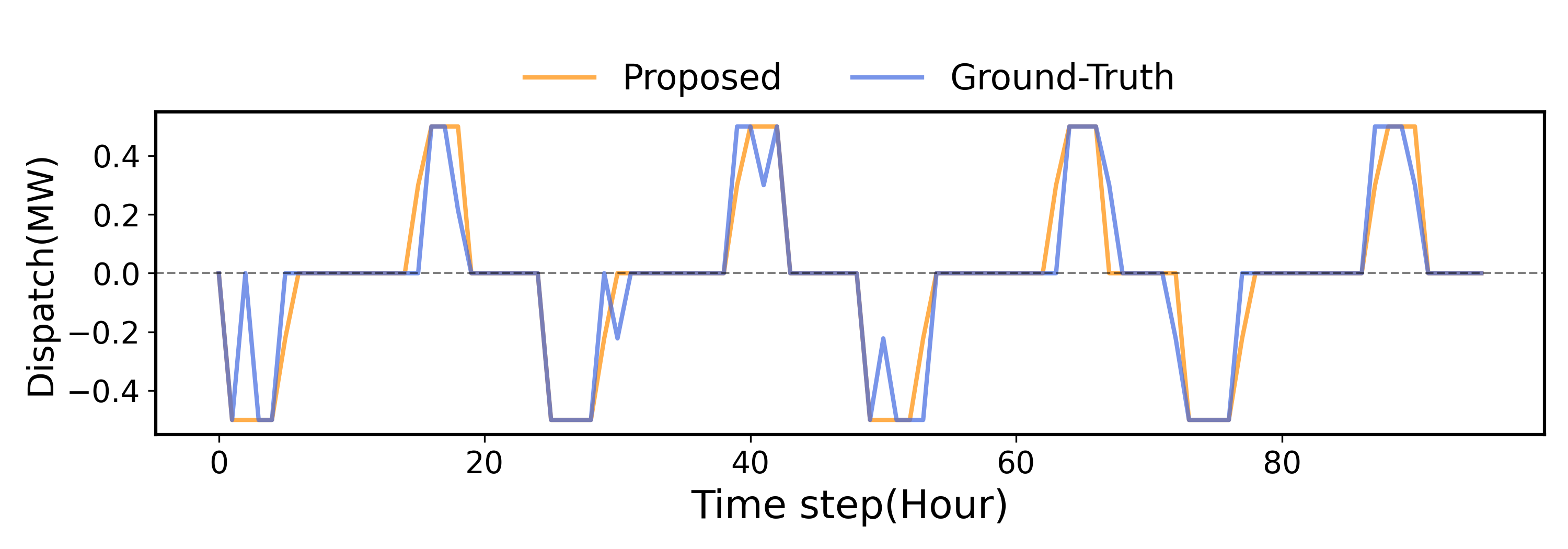} }
    \subfigure[]{ 
   \includegraphics[width = 1\linewidth]{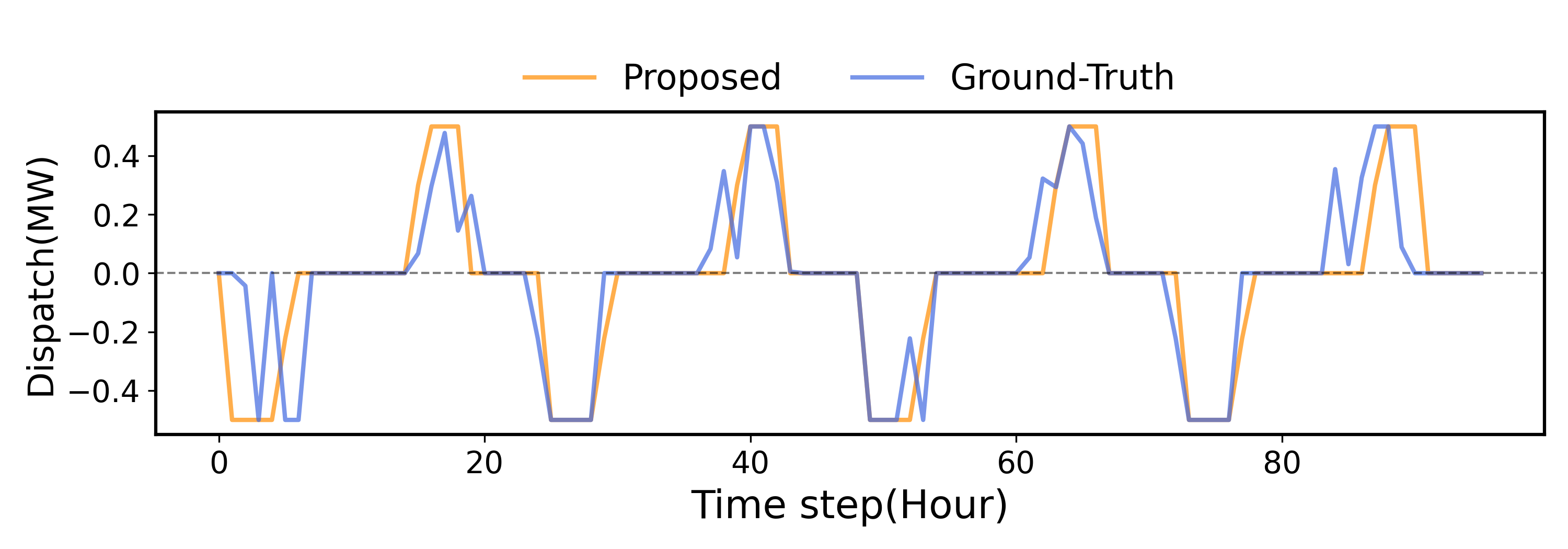} }
    	\caption{{Prediction performance of the proposed method without assuming knowledge of the cost term. Storage behavior with (a) linear cost; (b) linear and quadratic cost.}}
\label{Fig: compv}
\end{figure}

\subsubsection{Storage Behavior Prediction with Real-World Data}
We conduct experiments on public energy storage data in 2020 from Queensland University. The Tesla Powerpack battery
system has a 1.1 MW power rating and a 2.2 MWh capacity. The battery system is controlled by a demand response engine and engages in two major revenue activities: arbitrage and Frequency Control Ancillary Services (FCAS). According to its performance report \cite{WEH20}, the battery system only had a 2-hour duration of FCAS in 2020.  Therefore, energy storage behaviors are dominated by storage arbitrage. Some days did not have charging/discharging activities due to internet connection issues. Historical energy storage behaviors also have missing data issues for several days. These days are removed from the training and testing data. Six months of data are used for training, and three months of data are used for testing.  The data resolution is one sample every 30 minutes.  We include RTP data as input features. {The input features are structured as a $170 \times 96 \times 1$ tensor, where 170 represents the number of samples, 96 corresponds to half-hourly intervals, and 1 denotes the feature type. The target decision data is structured as a $170 \times 24$ matrix. The testing dataset follows the same format and contains 70 samples.}

The parameters of the energy storage arbitrage model are set as follows: power rating $P=1.11 MW$,  storage efficiency $\eta= 0.9$, storage capacity $E=2.2 MWh$, the prediction horizon $T=96$,   we assume the energy {storage behavior} has a linear cost term and a quadratic cost term  $u(p_t,b_t)=C_1p_t+C_2p_t^2$, where the coefficients of discharge cost is set as $C_1= 5$, $C_2= 3$. The hybrid loss parameter $\beta=0.01$.  
Note that predicting real-world energy storage behaviors is inherently challenging. First, complete knowledge of model parameters, such as prediction horizon and strategy, is unavailable. Second, our algorithm requires the SoC to be consistent across days, but data quality issues lead to some days being excluded, complicating the learning process. The stochastic nature of real-world energy storage adds further difficulty for accurate behavior predictions. Despite all these challenges, our method consistently outperforms the benchmark methods on this real-world dataset.

%the coefficients of charge cost $C3=0$, $C4=0$.
%The price data are very volatile in Australia and the number of training samples is limited. The two-stage method yields impractical results and is excluded from our comparison.

  We compare the proposed method with the first comparison algorithm, using two types of neural networks, LSTM and MLP. Fig. \ref{Fig: comp3} shows the prediction performance of the proposed method and the benchmark methods. The proposed method captures the trend and magnitude in ground-truth energy storage behaviors, albeit with some time deviations. Although the LSTM method predicts discharge decisions well, it consistently fails to predict the magnitude of charge decisions. The MLP method can approximately predict the timing of charge/discharge decisions, but it consistently struggles to predict their magnitude. Hence, the practical applicability of these two methods is highly constrained.  

Table \ref{exp3:conf_matrix} lists the confusion matrices and corresponding metrics. Table \ref{exp3:conf_matrix} shows that the proposed method achieves the best performance for all metrics. Fig. \ref{Fig:barcomp1}(c) shows the total dispatch energy for the proposed and comparison methods. Our method slightly over-predicts charge and discharge energy, while the two benchmark methods exhibit much higher prediction errors. These results underscore the substantial advantage of the proposed method.

\section{Conclusions}\label{con}
This paper introduces a decision-focused approach that integrates the physical energy {storage model} into a neural network-based architecture.  Our method incorporates prior knowledge from the storage arbitrage model and infers the hidden reward that guides energy storage decisions. We provide theoretical analysis for the perturbed loss function. To facilitate effective training, we propose a novel hybrid loss function. We demonstrate our proposed framework on two energy storage applications: the self-scheduling energy storage arbitrage problem and the energy storage behavior prediction problem.  Numerical results on both applications underscore the significant advantages of our proposed method. {In future work, we plan to extend our decision-focused framework to strategic energy storage operation under a price-maker scenario, incorporating a grid-based model and accounting for the storage system's influence on market prices. We will also explore incorporating grid losses \cite{SRL17} into our framework. Additionally, we will integrate our approach into behind-the-meter energy disaggregation for energy storage.}

 %\section*{Acknowledgment}
 
\bibliographystyle{IEEEtran} 
%\bibliography{./bibfiles/IEEEabrv,./bibfiles/ref,./bibfiles/MengWangPub}
%\bibliography{./bibfiles/IEEEabrv,./bibfiles/ref,./bibfiles/MengWangPub}
\bibliography{IEEEabrv,ref}
 
%\bibliographystyle{IEEEtran} 

%\bibliography{IEEEabrv,ref}

 %\iffalse
% \vspace{0.8in}
%\vspace{-1.5in}
%   \begin{IEEEbiography} [{\includegraphics[width=1.05in,height=1.25in,clip,keepaspectratio]{./Figures/Author_MingYi.pdf}}]	{Ming Yi} (M'23) received the B.E. degree in automation from Harbin Engineering University, Harbin, China, in 2016, and the M.S. degrees in control science and engineering from Harbin Institute of Technology, Harbin, China, in 2018, respectively. 

% He is currently a Ph.D. student in Rensselaer Polytechnic Institute, Troy, NY, USA. His research interests include signal processing, machine learning, power systems monitoring, and high-dimensional data analysis.
% \end{IEEEbiography}

%\fi
\clearpage
\newpage
%\vspace{-7in}
%\vspace{-1in}
%\fontsize{10pt}{10pt}\selectfont
%\setlength\columnsep{0.2in}  % Adjust column separation if needed
%\setlength\textwidth{7.1in}  
\section{Appendix}

\begin{algorithm} [!ht]
\caption{Self-Scheduling Energy Storage Arbitrage} 
\begin{algorithmic}[1]  \label{alg2}
\REQUIRE  The training data ${\mathcal{D}}=\{{\bm{x}}^1, {\bm{x}}^2,..., {\bm{x}}^n\}$; The testing data $\hat{\mathcal{D}}=\{(\hat{\bm{x}}^1, \hat{\bm{x}}^2,..., \hat{\bm{x}}^n\}$; The time horizon $T$; The power rating $P$ of energy storage; The SoC efficiency $\eta$ and capacity $E$; The initial SoC $e_0$;  The maximum epochs $M_{\max}$; The cost coefficients in $u(p_t, b_t)$.\\
\STATE 
\textbf{Initialization}: The weights $\bm{w}$ of predictor $g_w(\bm{x})$  are randomly initialized.\\
\textbf{Prepare optimal arbitrage decisions:}
\STATE{Given historical electricity price in ${\mathcal{D}}$, compute the optimal energy storage arbitrage decisions $\{\bar{\bm{y}}^1,\bar{\bm{y}}^2,...,\bar{\bm{y}}^n \}$, } and form the pairs $\mathcal{D}=\{(\bm{x}^1,\bar{\bm{y}}^1), (\bm{x}^2, \bar{\bm{y}}^2),...,(\bm{x}^n, \bar{\bm{y}}^n)\}$; \\
\textbf{Training stage:}
 \FOR{Epochs < $M_{\max}$ }
 \FOR{Each batch in $\mathcal{D}$ }
 \STATE{Predict the reward using the predictor $\hat{\bm{\lambda}}=g_w(\bm{x})$; }
  \STATE{Forward pass to compute the optimal decisions $p_t^*, b_t^*, t=1,2,...,T$ by equation (\ref{eq:arb});}
 \STATE{Compute the loss $\mathcal{L}_\epsilon(\hat{\bm{\lambda}},\bar{\bm{y}})$ by equation (\ref{eq:hybrid})};
  \STATE{Compute the total gradient by equation (\ref{sub:total})  and backward pass to update the weights $\bm{w}$ of predictor.};
 \\

\ENDFOR
\ENDFOR

%\STATE{\bm{Y}=\mathcal{H}^{\dagger}\bm{X}}

\textbf{Arbitrage stage:}

\STATE{Predict the reward using the predictor $\hat{\bm{\lambda}}=g_w(\hat{\bm{x}})$; }

 \STATE{Forward pass to compute the optimal schedules $\hat{p}_t, \hat{b}_t, t=1,2,...,T$ by equation (\ref{eq:arb});}
%\STATE{Compute the  by ($\ref{hruncertaintyindex}$)}

%\STATE{Compute the uncertainty index $U_{\textrm{index}}$ by  (\ref{Uncertaintyindex2})}
%\STATE{Compute the predictive mean and uncertainty index by  and , respectively}
\RETURN
The energy storage schedules  $(\hat{p}_t, \hat{b}_t)$ for $t=1,2,...T$. %The uncertainty index $U_{\textrm{index}}$.

\end{algorithmic} 
\end{algorithm}

\setcounter{theorem}{0}
\begin{proposition}\emph{The differentiablity of perturbed function.}
As noise $\bm{Z}$ is from Gaussian distribution, it has the density $\rho(\bm{Z})\propto \text{exp}(-\psi(\bm{Z}))$. For $\mathcal{R}_{\mathcal{X}}=\underset{{y\in \mathcal{X}}}{\text{max}}||\bm{y}_{\epsilon}^*(\hat{\bm{\lambda}})||$, we have
 \begin{itemize}
 \item $F_{\epsilon}(\hat{\bm{\lambda}}) $ is differentiable, and 
$\nabla_{\hat{\bm{\lambda}}} F_{\epsilon}(\hat{\bm{\lambda}})=\bm{y}_{\epsilon}^*(\hat{\bm{\lambda}}) =\mathbb{E}[ y^*(\hat{\bm{\lambda}}+\epsilon \bm{Z})]=\mathbb{E}[F(\hat{\bm{\lambda}}+\epsilon \bm{Z})\nabla_Z \psi(\bm{Z})/\epsilon]$. \\
 \item $F_{\epsilon}(\hat{\bm{\lambda}})$ is Lipschitz continuous, and  $||F_{\epsilon}(\hat{\bm{\lambda}})-F_{\epsilon}(\tilde{\bm{\lambda}})||
 \leq \mathcal{R}_{\mathcal{X}}|| \hat{\bm{\lambda}}-\tilde{\bm{\lambda}}||.$
\end{itemize}  
\end{proposition}

\begin{proof}
\noindent Based on the definition of $F_{\epsilon}(\hat{\bm{\lambda}})$, we have
\begin{equation} \label{proof: prop1}
\begin{aligned}
& F_{\epsilon}(\hat{\bm{\lambda}})=\mathbb{E}[ F(\hat{\bm{\lambda}}+\epsilon \bm{Z})] \\
& = \int \, F(\hat{\bm{\lambda}}+\epsilon \bm{Z}) \rho(\bm{Z}) d\bm{Z}  \\
& = \int \, F(\bm{G}) \rho(\frac{\bm{G}-\hat{\bm{\lambda}}}{\epsilon}) d\bm{G}  \\
\end{aligned}
\end{equation}

\noindent where $\bm{G}=\hat{\bm{\lambda}}+\epsilon \bm{Z}.$ We use the change of variable trick to obtain the last equation in \eqref{proof: prop1}.
As noise $\bm{Z}$ is from Gaussian distribution, $\psi(\bm{Z})$  is differentiable. Then we have $\nabla \rho(\bm{Z})=- \rho(\bm{Z})\nabla \psi(\bm{Z}) $. We apply the change of variable trick again and
\begin{equation} \label{eq: prop2}
\begin{aligned}
&\nabla_{\hat{\bm{\lambda}}}F_{\epsilon}(\hat{\bm{\lambda}}) \\
& = -\frac{1}{\epsilon}\int \, F(\bm{G}) \nabla_{\hat{\bm{\lambda}}} \rho(\frac{\bm{G}-\hat{\bm{\lambda}}}{\epsilon}) d\bm{G}  \\
& = \frac{1}{\epsilon}\int \, F(\bm{G}) \rho(\frac{\bm{G}-\hat{\bm{\lambda}}}{\epsilon})\nabla_{\hat{\bm{\lambda}}} \psi(\frac{\bm{G}-\hat{\bm{\lambda}}}{\epsilon}) d\bm{G}  \\
& = \frac{1}{\epsilon}\int \, F(\hat{\bm{\lambda}}+\epsilon \bm{Z}) \rho(\bm{Z})\nabla_Z \psi(\bm{Z}) d\bm{Z}  \\
&=\mathbb{E}[F(\hat{\bm{\lambda}}+\epsilon \bm{Z})\nabla_Z \psi(\bm{Z})/\epsilon]
\end{aligned}
\end{equation}

Next we show that $F_{\epsilon}(\hat{\bm{\lambda}})$ is Lipschitz continuous.
% \begin{equation} \label{eq: prop2}
% \begin{aligned}
% & ||F(\hat{\bm{\lambda}})-F(\tilde{\bm{\lambda}})||\\
% & =||\{{\textbf{max}}  (\hat{\bm{\lambda}})^Ty-u(y)\}-{\textbf{max}}\{  (\tilde{\bm{\lambda}} )^Ty-u(y)\}||\\
% & =|| \{\hat{\bm{\lambda}} ^Ty^*(\hat{\bm{\lambda}})-u(y^*(\hat{\bm{\lambda}}))\}- \{\tilde{\bm{\lambda}} ^Ty^*(\tilde{\bm{\lambda}})-u(y^*(\tilde{\bm{\lambda}}))\}||\\
% \end{aligned}
% \end{equation}
%A function f:R2→R is differentiable at (x0,y0)
% if its partial derivatives ∂f∂x
% and ∂f∂y
% are continuous at (x0,y0)
%.

From the mean value theorem,  we have

\begin{equation} \label{eq: prop2}
\begin{aligned}
& ||F_\epsilon(\hat{\bm{\lambda}})-F_\epsilon(\tilde{\bm{\lambda}})||\\
& \leq \underset{\hat{\bm{\lambda}}\in \mathbb{R}^T}{\text{sup}}||\nabla F_\epsilon(\hat{\bm{\lambda}})|| \, || \hat{\bm{\lambda}}-\tilde{\bm{\lambda}}|| \\
& =\underset{\hat{\bm{\lambda}}\in \mathbb{R}^T}{\text{sup}}||\bm{y}_{\epsilon}^*(\hat{\bm{\lambda}})|| \, || \hat{\bm{\lambda}}-\tilde{\bm{\lambda}}|| \\
& =\mathcal{R}_{\mathcal{X}} || \hat{\bm{\lambda}}-\tilde{\bm{\lambda}}|| \\
\end{aligned}
\end{equation}

%where $||\bm{y}_{\epsilon}^*(\hat{\bm{\lambda}})|| represents the charge/discharge decisions so 
Therefore, we can conclude that $F_\epsilon(\hat{\bm{\lambda}})$ is Lipschitz continuous.

Then we conclude the proof.
\end{proof}
\begin{proposition}\emph{Impact of scaling factor $\epsilon$.}
 \begin{itemize}
 \item $\bm{y}_{\epsilon}^*(\hat{\bm{\lambda}})$ is in the interior of $\mathcal{X}$. 
 \item For $\hat{\bm{\lambda}}$ such that $\bm{y}^*({\hat{\bm{\lambda}}})$ is a unique maximizer: when $\epsilon \rightarrow \infty$, $\bm{y}_{\epsilon}^*(\hat{\bm{\lambda}})\rightarrow \bm{y}_{1}^*(\frac{\hat{\bm{\lambda}}}{\epsilon})=\epsilon \mathbb{E}[ \underset{{\bm{y}\in \mathcal{X}}}{\textbf{argmax}} \{\bm{Z}^T\bm{y}\}]$; when  $\epsilon \rightarrow 0$, $\bm{y}_{\epsilon}^*(\hat{\bm{\lambda}})\rightarrow \bm{y}^*({\hat{\bm{\lambda}}}).$
\end{itemize}  
\end{proposition}

\begin{proof}

$\bm{y}_{\epsilon}^*(\hat{\bm{\lambda}}) =\mathbb{E}[ y^*(\hat{\bm{\lambda}}+\epsilon \bm{Z})]=\int \rho(\bm{Z})  y^*(\hat{\bm{\lambda}}+\epsilon \bm{Z})] d\bm{Z} = \underset{y\in \mathcal{Y}}{\sum} y P(y^*(\hat{\bm{\lambda}}+\epsilon \bm{Z})=y)$

\noindent where $P(y^*(\hat{\bm{\lambda}}+\epsilon \bm{Z})=y)$ denotes the probability of $y^*(\hat{\bm{\lambda}}+\epsilon \bm{Z})=y$. $\mathcal{Y}$ denotes the set of feasible solutions. 

Because $\rho(Z)$ is a positive density function, $P(y^*(\hat{\bm{\lambda}}+\epsilon \bm{Z})=y)$  is positive for all $y \in \mathcal{Y}$ and $\underset{y\in \mathcal{Y}}{\sum} P(y^*(\hat{\bm{\lambda}}+\epsilon \bm{Z})=y)=1$. 

Therefore, $\bm{y}_{\epsilon}^*(\hat{\bm{\lambda}})$ is in the interior of convex hull $\mathcal{X}$ of set $\mathcal{Y}$.

%$\bm{y}_{\epsilon}^*(\hat{\bm{\lambda}}) =\nabla_{\hat{\bm{\lambda}}} F_{\epsilon}(\hat{\bm{\lambda}})=\mathbb{E}[ \bm{y}^*(\hat{\bm{\lambda}}+\epsilon \bm{Z})]=\mathbb{E}[F(\hat{\bm{\lambda}}+\epsilon \bm{Z})\nabla_Z \psi(\bm{Z})/\epsilon]$

\begin{equation} \label{eq: prop2}
\begin{aligned}
& F_{\epsilon}(\hat{\bm{\lambda}})=\mathbb{E}[ F(\hat{\bm{\lambda}}+\epsilon \bm{Z})] \\
& =  \epsilon \mathbb{E}[ \underset{{\bm{y}\in \mathcal{X}}}{\textbf{max}} \{(\frac{\hat{\bm{\lambda}}}{\epsilon} +\bm{Z})^T\bm{y}-\frac{1}{\epsilon}u(\bm{y})\}] \\
& =\epsilon F_{1}(\frac{\hat{\bm{\lambda}}}{\epsilon}) \\
\end{aligned}
\end{equation}
\noindent where $F_{1}(\frac{\hat{\bm{\lambda}}}{\epsilon})$ denotes the perturbation without scaling by $\epsilon$. We set the subscript $F_{\epsilon}$ as $F_{1}$.

Because we have $\bm{y}_{\epsilon}^*(\hat{\bm{\lambda}})=\nabla_{\hat{\bm{\lambda}}} F_{\epsilon}(\hat{\bm{\lambda}})$ and  $F_{\epsilon}(\hat{\bm{\lambda}})= \epsilon F_{1}(\frac{\hat{\bm{\lambda}}}{\epsilon})$, we can obtain $\bm{y}_{\epsilon}^*(\hat{\bm{\lambda}})=\bm{y}_{1}^*(\frac{\hat{\bm{\lambda}}}{\epsilon})$.
When $\epsilon \rightarrow \infty$,

\begin{equation} \label{eq: prop2}
\begin{aligned}
 F_{\epsilon}(\hat{\bm{\lambda}})
& =  \epsilon  \mathbb{E}[\underset{{\bm{y}\in \mathcal{X}}}{\textbf{max}} \{(\frac{\hat{\bm{\lambda}}}{\epsilon} +\bm{Z})^T\bm{y}-\frac{1}{\epsilon}u(\bm{y})\}] \\
& =   \epsilon \mathbb{E}[ \underset{{\bm{y}\in \mathcal{X}}}{\textbf{max}} \{\bm{Z}^T\bm{y}\}] \\
\end{aligned}
\end{equation}

Thus, $\epsilon \rightarrow \infty$, $\bm{y}_{\epsilon}^*(\hat{\bm{\lambda}})\rightarrow \bm{y}_{1}^*(\frac{\hat{\bm{\lambda}}}{\epsilon})=\epsilon \mathbb{E}[ \underset{{\bm{y}\in \mathcal{X}}}{\textbf{argmax}} \{\bm{Z}^T\bm{y}\}]$

Next we show that when  $\epsilon \rightarrow 0$, $\bm{y}_{\epsilon}^*(\hat{\bm{\lambda}})\rightarrow \bm{y}^*({\hat{\bm{\lambda}}}).$
\\
We define $\epsilon\Omega(\bm{y})$ be the Fenchal dual of $F{(\hat{\bm{\lambda}})}$

\begin{equation} \label{eq: prop2}
\begin{aligned}
&F{({\hat{\bm{\lambda}}})} =  \underset{{\bm{y}\in \mathcal{X}}}{\textbf{sup}} \{{\hat{\bm{\lambda}}} ^T\bm{y}-\epsilon\Omega(\bm{y}) \}
\end{aligned}
\end{equation}

Thus, 

\begin{equation} \label{eq: prop2}
\begin{aligned}
y^*({\hat{\bm{\lambda}}})=\underset{{\bm{y}\in \mathcal{X}}}{\textbf{argmax}} \{{\hat{\bm{\lambda}}}^T\bm{y}-  \epsilon\Omega(\bm{y})\}
\end{aligned}
\end{equation}

We define $\Gamma(\bm{y})$ be the Fenchal dual of $F_1{(\frac{\hat{\bm{\lambda}}}{\epsilon})}$

\begin{equation} \label{eq: prop2}
\begin{aligned}
&F_1{(\frac{\hat{\bm{\lambda}}}{\epsilon})} =  \underset{{\bm{y}\in \mathcal{X}}}{\textbf{sup}} \{(\frac{\hat{\bm{\lambda}}}{\epsilon} )^T\bm{y}-\Gamma(\bm{y}) \}
\end{aligned}
\end{equation}

Then $\epsilon \Gamma(\bm{y})$ is the Fenchal dual of $F_\epsilon{(\hat{\bm{\lambda}})}$
\begin{equation} \label{eq: prop2}
\begin{aligned}
F_\epsilon({\hat{\bm{\lambda}}})=\epsilon F_1{(\frac{\hat{\bm{\lambda}}}{\epsilon})} &=  \underset{{\bm{y}\in \mathcal{X}}}{\textbf{sup}} \{{\hat{\bm{\lambda}}}^T\bm{y}- \epsilon \Gamma(\bm{y})\}\\
\end{aligned}
\end{equation}

Therefore, 

\begin{equation} \label{eq: prop2}
\begin{aligned}
y_\epsilon^*({\hat{\bm{\lambda}}})=\underset{{\bm{y}\in \mathcal{X}}}{\textbf{argmax}} \{{\hat{\bm{\lambda}}}^T\bm{y}- \epsilon \Gamma(\bm{y})\}
\end{aligned}
\end{equation}

From inequality \eqref{proof1: step4}, we  have $F_\epsilon({\hat{\bm{\lambda}}}) \geq  F({\hat{\bm{\lambda}}})$, then
\begin{equation} \label{proof2:inequality}
\begin{aligned}
& {\hat{\bm{\lambda}}}^T\bm{y}_\epsilon^*(\hat{\bm{\lambda}})- \epsilon \Gamma(\bm{y}_\epsilon^*(\hat{\bm{\lambda}})) \geq  {\hat{\bm{\lambda}}}^T\bm{y}^*(\hat{\bm{\lambda}})- \epsilon \Omega(\bm{y}^*(\hat{\bm{\lambda}}))
\end{aligned}
\end{equation}

By rearranging inequality \eqref{proof2:inequality}, we obtain
\begin{equation}  \label{proof:eq36}
\begin{aligned}
&{\hat{\bm{\lambda}}}^T\bm{y}^*(\hat{\bm{\lambda}})- {\hat{\bm{\lambda}}}^T\bm{y}_\epsilon^*(\hat{\bm{\lambda}})   \leq \epsilon (\Omega(\bm{y}^*(\hat{\bm{\lambda}}))-  \Gamma(\bm{y}_\epsilon^*(\hat{\bm{\lambda}})) )
\end{aligned}
\end{equation}

As $\bm{y}_{\epsilon}^*(\hat{\bm{\lambda}})$ is in the interior of $\mathcal{X}$,  we have

\begin{equation} \label{proof:eq37}
\begin{aligned}
&  {\hat{\bm{\lambda}}}^T\bm{y}^*(\hat{\bm{\lambda}})- {\hat{\bm{\lambda}}}^T\bm{y}_\epsilon^*(\hat{\bm{\lambda}})  \geq 0
\end{aligned}
\end{equation}

By combining inequalities \eqref{proof:eq36} and \eqref{proof:eq37}, we obtain

\begin{equation} 
\begin{aligned}
& 0 \leq {\hat{\bm{\lambda}}}^T\bm{y}^*(\hat{\bm{\lambda}})- {\hat{\bm{\lambda}}}^T\bm{y}_\epsilon^*(\hat{\bm{\lambda}})   \leq \epsilon (\Omega(\bm{y}^*(\hat{\bm{\lambda}}))-  \Gamma(\bm{y}_\epsilon^*(\hat{\bm{\lambda}})) )
\end{aligned}
\end{equation}

Because $\Omega(\bm{y})$ and $\Gamma(\bm{y})$ are continuous functions,  they are bounded on $\mathcal{X}$. Therefore, $\epsilon (\Omega(\bm{y}^*(\hat{\bm{\lambda}}))-  \Gamma(\bm{y}_\epsilon^*(\hat{\bm{\lambda}})) )$ is bounded. Consequently, as  $\epsilon \rightarrow 0$, we have ${\hat{\bm{\lambda}}}^T\bm{y}_\epsilon^*(\hat{\bm{\lambda}}) \rightarrow {\hat{\bm{\lambda}}}^T\bm{y}^*(\hat{\bm{\lambda}})$. Since   $\mathcal{X}$ is closed and bounded, for any sequence $\epsilon_n \rightarrow 0$,  the corresponding sequence $y_{\epsilon_n}^*(\hat{\bm{\lambda}})$ is within $\mathcal{X}$. Thus, there exists a sub-sequence $y_{\epsilon_{n_k}}(\hat{\bm{\lambda}})$ that converges to some limit $\underset{{k \to \infty}}{\lim} y_{\epsilon_{n_k}}(\hat{\bm{\lambda}}) \rightarrow \bm{y}_{L}$ and the limit $\bm{y}_{L} \in \mathcal{X}.$   Since ${\hat{\bm{\lambda}}}^T\bm{y}_{\epsilon_{n_k}}^*(\hat{\bm{\lambda}}) \rightarrow {\hat{\bm{\lambda}}}^T\bm{y}^*(\hat{\bm{\lambda}})$, it follows that ${\hat{\bm{\lambda}}}^T\bm{y}_{L} = {\hat{\bm{\lambda}}}^T\bm{y}^*(\hat{\bm{\lambda}})$. Given that $\bm{y}^*(\hat{\bm{\lambda}})$ is the unique solution, we have $\bm{y}^*(\hat{\bm{\lambda}})=\bm{y}_{L}$. Thus, all subsequences $y_{\epsilon_{n_k}}(\hat{\bm{\lambda}})$ converge to the same limit $\bm{y}^*(\hat{\bm{\lambda}})$. Therefore,  when  $\epsilon \rightarrow 0$, $\bm{y}_{\epsilon}^*(\hat{\bm{\lambda}})\rightarrow \bm{y}^*({\hat{\bm{\lambda}}}).$ This concludes the proof.
\end{proof}

\begin{proposition}\emph{The property of perturbed loss function $\mathcal{L}_\epsilon^{PFY}(\hat{\bm{\lambda}},\bar{\bm{y}} )$.}
 \begin{itemize}
 \item $\mathcal{L}_\epsilon^{PFY}(\hat{\bm{\lambda}},\bar{\bm{y}} )$ is a  convex function of  ${\hat{\bm{\lambda}}}$.
 \item $ 0\leq \mathcal{L}_\epsilon^{PFY}(\hat{\bm{\lambda}},\bar{\bm{y}} ) -\mathcal{L}^{FY}(\hat{\bm{\lambda}},\bar{\bm{y}})  \leq TP \epsilon$.
\end{itemize}  
\end{proposition}

%\noindent \textbf{Proposition 6}
%$ 0 \leq \mathcal{L}^{FY}(\hat{\bm{\lambda}},\bar{\bm{y}})-\mathcal{L}_\epsilon^{PFY}(\hat{\bm{\lambda}},\bar{\bm{y}}) \leq C \epsilon $

\begin{proof}

$\mathcal{L}^{FY}(\hat{\bm{\lambda}},\bar{\bm{y}})$ is a convex function of $\hat{\bm{\lambda}}$ because  it is a maximum of linear function of $\hat{\bm{\lambda}}$. 
Based on the definition of $\mathcal{L}_\epsilon^{PFY}(\hat{\bm{\lambda}},\bar{\bm{y}} )$, and for any $\mu \in [0,1]$, $\hat{\bm{\lambda}}^1$, $\hat{\bm{\lambda}}^2 \in \mathbb{R}^T$,  we have
\begin{equation} \label{proof1: step0}
\begin{aligned}
&\mu \mathcal{L}_\epsilon^{PFY}(\hat{\bm{\lambda}}^1,\bar{\bm{y}} )+(1-\mu)\mathcal{L}_\epsilon^{PFY}(\hat{\bm{\lambda}}^2,\bar{\bm{y}} ) \\&=\mu\mathbb{E}[\mathcal{L}^{FY}(\hat{\bm{\lambda}}^1+\epsilon \bm{Z},\bar{\bm{y}})]+(1-\mu)\mathbb{E}[\mathcal{L}^{FY}(\hat{\bm{\lambda}}^2+\epsilon \bm{Z},\bar{\bm{y}})]\\
&=\mathbb{E}[\mu\mathcal{L}^{FY}(\hat{\bm{\lambda}}^1+\epsilon \bm{Z},\bar{\bm{y}})+(1-\mu)\mathcal{L}^{FY}(\hat{\bm{\lambda}}^2+\epsilon \bm{Z},\bar{\bm{y}})]\\&
\geq \mathbb{E}[\mathcal{L}^{FY}(\mu\hat{\bm{\lambda}}^1+(1-\mu)\hat{\bm{\lambda}}^2+\epsilon \bm{Z},\bar{\bm{y}})]\\
& = \mathcal{L}_\epsilon^{PFY}(\hat{\bm{\lambda}}^\mu,\bar{\bm{y}} )
\end{aligned}
\end{equation}
where $\hat{\bm{\lambda}}^\mu=\mu \hat{\bm{\lambda}}^1+(1-\mu)\hat{\bm{\lambda}}^2$.
The expectation preserves the convexity. Therefore, we conclude that $\mathcal{L}_\epsilon^{PFY}(\hat{\bm{\lambda}},\bar{\bm{y}} ) = \mathbb{E}[\mathcal{L}^{FY}(\hat{\bm{\lambda}}+\epsilon \bm{Z},\bar{\bm{y}})]$  is a convex function of $\hat{\bm{\lambda}}$. \\
\noindent For all $\bm{Z}$, we have
\begin{equation} \label{proof1: step1}
\begin{aligned}
&\underset{{\bm{y} \in \mathcal{X}}}{\textbf{max}}  \{(\hat{\bm{\lambda}} +\epsilon \bm{Z})^T\bm{y}-u(\bm{y})\} \geq   (\hat{\bm{\lambda}} +\epsilon \bm{Z})^T\bm{y}-u(\bm{y})
\end{aligned}
\end{equation}

\noindent We take the expectation over $\bm{Z}$'s distribution on both side of \eqref{proof1: step1}, and we can get

\begin{equation} \label{proof: prop3}
\begin{aligned}
&\mathbb{E}[\underset{{\bm{y} \in \mathcal{X}}}{\textbf{max}} \{(\hat{\bm{\lambda}} +\epsilon \bm{Z})^T\bm{y}-u(\bm{y})\}] \geq   \mathbb{E}[(\hat{\bm{\lambda}} +\epsilon \bm{Z})^T \bm{y}-u(\bm{y}) ] 
\end{aligned}
\end{equation}

\noindent For the right hand side of \eqref{proof: prop3},

\begin{equation} \label{proof1: step3}
\begin{aligned}  
&\mathbb{E}[(\hat{\bm{\lambda}} +\epsilon \bm{Z})^T\bm{y}-u(\bm{y}) ] = (\hat{\bm{\lambda}}+\epsilon \mathbb{E}[\bm{Z}] )^T\bm{y} -u(\bm{y}) \\&= \hat{\bm{\lambda}}^T\bm{y} -u(\bm{y})\\
\end{aligned}
\end{equation}

The last equation in \eqref{proof1: step3} holds because  $\bm{Z}$ is from a Gaussian distribution $\mathcal{N}(\textbf{0},\textbf{1})$ with $\mathbb{E}[\bm{Z}]=\textbf{0}$.

\noindent Then inequality \eqref{proof: prop3} can be written as
\begin{equation} \label{proof1: step4}
\begin{aligned}
&\mathbb{E}[\underset{{\bm{y} \in \mathcal{X}}}{\textbf{max}} \{(\hat{\bm{\lambda}} +\epsilon \bm{Z})^T\bm{y}-u(\bm{y})\}] \geq  \hat{\bm{\lambda}}^T \bm{y}-u(\bm{y})  
\end{aligned}
\end{equation}

\noindent We subtract $\hat{\bm{\lambda}}^T \bar{\bm{y}}  -u(\bar{\bm{y}})$ on both sides,

 \begin{equation} \label{proof1: step5}
\begin{aligned}
&\mathbb{E}[\underset{{\bm{y} \in \mathcal{X}}}{\textbf{max}} \{(\hat{\bm{\lambda}} +\epsilon \bm{Z})^T\bm{y}-u(\bm{y})\}] -( \hat{\bm{\lambda}}^T \bar{\bm{y}} 
 -u(\bar{\bm{y}})) \\&
 \geq    \hat{\bm{\lambda}}^T\bm{y}-u(\bm{y}) -( \hat{\bm{\lambda}}^T \bar{\bm{y}}
 -u(\bar{\bm{y}})) 
\end{aligned}
\end{equation}

\noindent It is clear that the left hand side is $ \mathcal{L}_\epsilon^{PFY}(\hat{\bm{\lambda}},\bar{\bm{y}} )$ and the right hand side is $\mathcal{L}^{FY}(\hat{\bm{\lambda}},\bar{\bm{y}})$. Thus, we conclude that $\mathcal{L}_\epsilon^{PFY}(\hat{\bm{\lambda}},\bar{\bm{y}} ) \geq \mathcal{L}^{FY}(\hat{\bm{\lambda}},\bar{\bm{y}}).$

For all $\bm{Z} \in \mathbb{R}^T$, we have
\begin{equation} \label{proof1: step5}
\begin{aligned}
&\underset{{\bm{y} \in \mathcal{X}}}{\textbf{max}} \{(\hat{\bm{\lambda}} +\epsilon \bm{Z})^T\bm{y}-u(\bm{y})\} \\&
 \leq   \underset{{\bm{y} \in \mathcal{X}}}{\textbf{max}} \{ \hat{\bm{\lambda}}^T\bm{y} -u(\bm{y})\} +\underset{{\bm{y} \in \mathcal{X}}}{\textbf{max}} \{ \epsilon \bm{Z}^T\bm{y}\}\\
 &=\underset{{\bm{y} \in \mathcal{X}}}{\textbf{max}} \{ \hat{\bm{\lambda}}^T\bm{y} -u(\bm{y})\} + \epsilon \underset{{\bm{y} \in \mathcal{X}}}{\textbf{max}} \{ \bm{Z}^T\bm{y}\} \\
 %&=F(\hat{\bm{\lambda}})+
\end{aligned}
\end{equation}

Take the expectation on both sides,

\begin{equation} \label{proof1: step5}
\begin{aligned}
&\mathbb{E}[\underset{{\bm{y} \in \mathcal{X}}}{\textbf{max}} \{(\hat{\bm{\lambda}} +\epsilon \bm{Z})^T\bm{y}-u(\bm{y})\}] \\&
 \leq  \underset{{\bm{y} \in \mathcal{X}}}{\textbf{max}} \{ \hat{\bm{\lambda}}^T\bm{y} -u(\bm{y})\} + \epsilon \mathbb{E}[\underset{{\bm{y} \in \mathcal{X}}}{\textbf{max}} \{ \bm{Z}^T\bm{y}\} ]\\
\end{aligned}
\end{equation}

\noindent We subtract $\hat{\bm{\lambda}}^T \bar{\bm{y}}
 -u(\bar{\bm{y}})$ on both sides,

 \begin{equation} \label{proof1: step5}
\begin{aligned}
&\mathbb{E}[\underset{{\bm{y} \in \mathcal{X}}}{\textbf{max}} \{(\hat{\bm{\lambda}} +\epsilon \bm{Z})^T\bm{y}-u(\bm{y})\}] -( \hat{\bm{\lambda}}^T \bar{\bm{y}} 
 -u(\bar{\bm{y}})) \leq \\&\underset{{\bm{y} \in \mathcal{X}}}{\textbf{max}} \{ \hat{\bm{\lambda}}^T\bm{y}-u(\bm{y}) \}-( \hat{\bm{\lambda}}^T \bar{\bm{u}}
 -u(\bar{\bm{y}})) +\epsilon \mathbb{E}[\underset{{\bm{y} \in \mathcal{X}}}{\textbf{max}} \{ \bm{Z}^T\bm{y}\} ]
\end{aligned}
\end{equation}

Because $\mathbb{E}[\underset{\bm{y}\in \mathcal{X}}{\textbf{max}}  \{  \bm{Z}^T\bm{y}\} ] \leq TP$, we conclude that $\mathcal{L}_\epsilon^{PFY}(\hat{\bm{\lambda}},\bar{\bm{y}} ) \leq \mathcal{L}^{FY}(\hat{\bm{\lambda}},\bar{\bm{y}})+ TP\epsilon$. Thus, when $\epsilon \rightarrow 0$, combining two inequalities lead to $\mathcal{L}_\epsilon^{PFY}(\hat{\bm{\lambda}},\bar{\bm{y}}) \rightarrow \mathcal{L}^{FY}(\hat{\bm{\lambda}},\bar{\bm{y}})$
\end{proof}

\begin{proposition}\emph{Gradient of perturbed loss function.}

 \begin{itemize}
 \item The gradient of the perturbed loss function is
  \begin{equation} \label{}
\begin{aligned}
  & {\nabla \mathcal{L}_\epsilon^{PFY}(\hat{\bm{\lambda}},\bar{\bm{y}})} =\bm{y}_\epsilon^*(\hat{\bm{\lambda}})-\bar{\bm{y}}. 
\end{aligned}
\end{equation}
 \item ${\nabla \mathcal{L}_\epsilon^{PFY}(\hat{\bm{\lambda}},\bar{\bm{y}})} $ is Lipschitz continuous.
\end{itemize}  
\end{proposition}

%\noindent \textit{Proof of Proposition 7.}

%\noindent \textbf{Proposition 4}\\
%For $\sqrt{T}P=\underset{{y\in \mathcal{X}}}{\text{max}}||\bm{y}_{\epsilon}^*(\hat{\bm{\lambda}})||$ and $\mathbb{E}[||\psi(\bm{Z})/\epsilon]|^2]^{\frac{1}{2}}=M$,  $F_{\epsilon}(\hat{\bm{\lambda}})$ is Lipschitz continuous.
%${\nabla \mathcal{L}_\epsilon^{PFY}(\hat{\bm{\lambda}},\bar{\bm{y}})} $ is Lipschitz continuous.

\begin{proof}
% \begin{equation} \label{eq: prop2}
% \begin{aligned}
% & ||F_\epsilon(\hat{\bm{\lambda}})-F_\epsilon(\tilde{\bm{\lambda}})||\\
% & =||\mathbb{E}[F(\hat{\bm{\lambda}}+\epsilon \bm{Z})-F(\tilde{\bm{\lambda}}+\epsilon \bm{Z})]||\\
% & =||\mathbb{E}[\{{\textbf{max}}  (\hat{\bm{\lambda}} +\epsilon \bm{Z})^Ty-u(y)\}-{\textbf{max}}\{  (\tilde{\bm{\lambda}} +\epsilon \bm{Z})^T-u(y)\}||\\
% & =||(  (\hat{\bm{\lambda}} +\epsilon \bm{Z})^T-{\textbf{max}}  (\tilde{\bm{\lambda}} +\epsilon \bm{Z}))^Ty||\\
% \end{aligned}
% \end{equation}
\begin{equation}\label{}
\begin{aligned}
 &\mathcal{L}^{PFY}_{\epsilon}(\tilde{\bm{\lambda}},\bar{\bm{y}}) =\\
 &\underset{{\bm{y} \in \mathcal{X}}}{\textbf{max} }     \{ (\tilde{\bm{\lambda}}+\epsilon \bm{Z} )^T\bm{y}-u(\bm{y})\}- (\tilde{\bm{\lambda}}^T \bar{\bm{y}}  -u(\bar{\bm{y}}))\\
 & \geq \underset{{\bm{y} \in \mathcal{X}}}{\textbf{max} }  \{ (\hat{\bm{\lambda}}+\epsilon \bm{Z})^T\bm{y}-u(\bm{y})\}+ {\bm{y}_\epsilon^{*}(\hat{\bm{\lambda}})}^T(\tilde{\bm{\lambda}}-\hat{\bm{\lambda}}) \\
 &-( \tilde{\bm{\lambda}}^T\bar{\bm{y}}-u(\bar{\bm{y}}))\\
 & = \underset{{\bm{y} \in \mathcal{X}}}{\textbf{max} }     \{ (\hat{\bm{\lambda}}+\epsilon \bm{Z})^T\bm{y}-u(\bm{y})\}-(\hat{\bm{\lambda}}^T \bar{\bm{y}}-u(\bar{\bm{y}}))\\&+ \bm{y}_\epsilon^{*}(\hat{\bm{\lambda}})^T(\tilde{\bm{\lambda}}-\hat{\bm{\lambda}})+(\hat{\bm{\lambda}}^T \bar{\bm{y}}-u(\bar{\bm{y}}))- (\tilde{\bm{\lambda}}^T \bar{\bm{y}}
 -u(\bar{\bm{y}}))\\
  & = \underset{{\bm{y} \in \mathcal{X}}}{\textbf{max} }     \{ (\hat{\bm{\lambda}}+\epsilon \bm{Z})^T\bm{y}-u(\bm{y})\}-(\hat{\bm{\lambda}}^T \bar{\bm{y}}-u(\bar{\bm{y}}))\\&+ \bm{y}_\epsilon^{*}(\hat{\bm{\lambda}})^T (\tilde{\bm{\lambda}}-\hat{\bm{\lambda}})+(\hat{\bm{\lambda}}^T \bar{\bm{y}})- (\tilde{\bm{\lambda}}^T \bar{\bm{y}})\\
    & = \underset{{\bm{y} \in \mathcal{X}}}{\textbf{max} }     \{ (\hat{\bm{\lambda}}+\epsilon \bm{Z})^T\bm{y}-u(\bm{y})\}-(\hat{\bm{\lambda}}^T \bar{\bm{y}}-u(\bar{\bm{y}}))\\&+ (\bm{y}_\epsilon^*(\hat{\bm{\lambda}})-\bar{\bm{y}})^T(\tilde{\bm{\lambda}}-\hat{\bm{\lambda}})
\end{aligned}
\end{equation}

\noindent where the inequality follows $\bm{y}_\epsilon^*(\hat{\bm{\lambda}})\in {\textbf{argmax} }     \{(\hat{\bm{\lambda}}+\epsilon \bm{Z})^T\bm{y}-u(\bm{y})\} =\frac{\partial \mathcal{L}^1_\epsilon(\hat{\bm{\lambda}},\bar{\bm{y}} )}{\partial \hat{\bm{\lambda}}} $.  To simplify the notation, we denote $\underset{{\bm{y} \in \mathcal{X}}}{\textbf{max} }     \{ (\tilde{\bm{\lambda}} +\epsilon \bm{Z})^T\bm{y}-u(\bm{y})\}$ as  $\mathcal{L}_{\epsilon}^{1}(\tilde{\bm{\lambda}},{\bm{p}},{\bm{b}})$, and we have
$\mathcal{L}_{\epsilon}^{1}(\tilde{\bm{\lambda}},\bar{\bm{y}})\geq \mathcal{L}_{\epsilon}^{1}(\hat{\bm{\lambda}},\bar{\bm{y}})+ \frac{\partial \mathcal{L}_\epsilon^{1}(\hat{\bm{\lambda}},\bar{\bm{y}})}{\partial \hat{\bm{\lambda}}}^T(\tilde{\bm{\lambda}}-\hat{\bm{\lambda}})$.

\noindent Then we conclude that $\bm{y}_\epsilon^*(\hat{\bm{\lambda}}) -\bar{\bm{y}}$ is a subgradient of $\mathcal{L}_\epsilon^{PFY}(\hat{\bm{\lambda}},\bar{\bm{y}} )$, i.e., $\bm{y}_\epsilon^*(\hat{\bm{\lambda}}) -\bar{\bm{y}} \in \frac{\partial \mathcal{L}_\epsilon^{PFY}(\hat{\bm{\lambda}},\bar{\bm{y}} )}{\partial \hat{\bm{\lambda}}}$. Because $\mathcal{L}_\epsilon^{PFY}(\hat{\bm{\lambda}},\bar{\bm{y}} )$ is smooth, $\bm{y}_\epsilon^*(\hat{\bm{\lambda}}) -\bar{\bm{y}}$ is the gradient of $\mathcal{L}_\epsilon^{PFY}(\hat{\bm{\lambda}},\bar{\bm{y}} )$.
\\

Next, we show ${\nabla \mathcal{L}_\epsilon^{PFY}(\hat{\bm{\lambda}},\bar{\bm{y}} )}$ is Lipschitz continuous. 

\begin{equation} \label{eq: prop2}
\begin{aligned}
&{\nabla \mathcal{L}_\epsilon^{PFY}(\hat{\bm{\lambda}},\bar{\bm{y}} )} - {\nabla \mathcal{L}_\epsilon^{PFY}(\tilde{\bm{\lambda}},\bar{\bm{y}} )} \\
& =(\bm{y}_{\epsilon}^*(\hat{\bm{\lambda}})-\bar{\bm{y}})- (\bm{y}_{\epsilon}^*(\tilde{\bm{\lambda}})-\bar{\bm{y}})\\
&=\bm{y}_{\epsilon}^*(\hat{\bm{\lambda}})- \bm{y}_{\epsilon}^*(\tilde{\bm{\lambda}}) \\
&=\mathbb{E}[F(\hat{\bm{\lambda}}+\epsilon \bm{Z})\psi(\bm{Z})/\epsilon]-\mathbb{E}[F(\tilde{\bm{\lambda}}+\epsilon \bm{Z})\psi(\bm{Z})/\epsilon]\\
&=\mathbb{E}[(F(\hat{\bm{\lambda}}+\epsilon \bm{Z})-F(\tilde{\bm{\lambda}}+\epsilon \bm{Z}))\psi(\bm{Z})/\epsilon]\\
\end{aligned}
\end{equation}

Based on Cauchy–Schwarz inequality, we can obtain

\begin{equation} \label{eq: prop2}
\begin{aligned}
&||{\nabla \mathcal{L}_\epsilon^{PFY}(\hat{\bm{\lambda}},\bar{\bm{y}} )} - {\nabla \mathcal{L}_\epsilon^{PFY}(\tilde{\bm{\lambda}},\bar{\bm{y}} )}|| \\
&=||\mathbb{E}[(F(\hat{\bm{\lambda}}+\epsilon \bm{Z})-F(\tilde{\bm{\lambda}}+\epsilon \bm{Z}))\psi(\bm{Z})/\epsilon]||\\
& \leq ||\mathbb{E}[(F(\hat{\bm{\lambda}}+\epsilon \bm{Z})-F(\tilde{\bm{\lambda}}+\epsilon \bm{Z}))||^2]^{\frac{1}{2}} \mathbb{E}[||\psi(\bm{Z})||^2/\epsilon^2]^{\frac{1}{2}}\\
& \leq \mathcal{R}_{\mathcal{X}} ||\hat{\bm{\lambda}}-\tilde{\bm{\lambda}}|| \mathbb{E}[||\psi(\bm{Z})||^2/\epsilon^2]^{\frac{1}{2}}\\
& = \frac{\mathcal{R}_{\mathcal{X}}M}{\epsilon} ||\hat{\bm{\lambda}}-\tilde{\bm{\lambda}}||\\
\end{aligned}
\end{equation}

Therefore, ${\nabla \mathcal{L}_\epsilon^{PFY}(\hat{\bm{\lambda}},\bar{\bm{y}} )}$ is $\frac{\mathcal{R}_{\mathcal{X}}M}{\epsilon}$-Lipschitz continuous.

\end{proof}

\noindent \textit{Convergence Analysis of the Optimization Layer}

In the main text, we illustrate the difficulty of analyzing the convergence of the whole pipline. Thus, we focus on the perturbed optimization layer and consider $\mathcal{L}_\epsilon^{PFY}(\hat{\bm{\lambda}},\bar{\bm{y}})$ as a function of $\hat{\bm{\lambda}}$. Given that $\mathcal{L}_\epsilon^{PFY}(\hat{\bm{\lambda}},\bar{\bm{y}} )$ is a  convex function of  ${\hat{\bm{\lambda}}}$ in Proposition 3, and ${\nabla \mathcal{L}_\epsilon^{PFY}(\hat{\bm{\lambda}},\bar{\bm{y}})} $ is Lipschitz continuous in Proposition 4, the proof of convergence for the optimization layer is straightforward.
\begin{proof}

As ${\nabla \mathcal{L}_\epsilon^{PFY}(\hat{\bm{\lambda}},\bar{\bm{y}})} $ is $\frac{\mathcal{R}_{\mathcal{X}}M}{\epsilon}$-Lipschitz continuous, it implies that ${\nabla^2 \mathcal{L}_\epsilon^{PFY}(\hat{\bm{\lambda}},\bar{\bm{y}})} 		\preceq  \frac{\mathcal{R}_{\mathcal{X}}M}{\epsilon} \bm{I} $. Take the quadratic expansion over ${ \mathcal{L}_\epsilon^{PFY}(\hat{\bm{\lambda}},\bar{\bm{y}} )} $ and obtain

\begin{equation} \label{eq: convergence1}
\resizebox{0.48\textwidth}{!}{
\begin{math}
\begin{aligned}
&{ \mathcal{L}_\epsilon^{PFY}(\tilde{\bm{\lambda}},\bar{\bm{y}} )}  \leq { \mathcal{L}_\epsilon^{PFY}(\hat{\bm{\lambda}},\bar{\bm{y}} )} +\nabla { \mathcal{L}_\epsilon^{PFY}(\hat{\bm{\lambda}},\bar{\bm{y}} )}^T(\tilde{\bm{\lambda}}-\hat{\bm{\lambda}})+\\
&\frac{1}{2} \nabla^2 { \mathcal{L}_\epsilon^{PFY}(\hat{\bm{\lambda}},\bar{\bm{y}} )}||\tilde{\bm{\lambda}}-\hat{{\bm{\lambda}}} ||_2^2 \\
&\leq { \mathcal{L}_\epsilon^{PFY}(\hat{\bm{\lambda}},\bar{\bm{y}} )} +\nabla { \mathcal{L}_\epsilon^{PFY}(\hat{\bm{\lambda}},\bar{\bm{y}} )}^T(\tilde{\bm{\lambda}}-\hat{\bm{\lambda}})+\frac{\mathcal{R}_{\mathcal{X}}M}{2\epsilon}||\tilde{\bm{\lambda}}-\hat{{\bm{\lambda}}} ||_2^2 \\
\end{aligned}
\end{math}
}
\end{equation}

\noindent  Let $\tilde{\bm{\lambda}}= \hat{\bm{\lambda}}- \alpha \nabla{ \mathcal{L}_\epsilon^{PFY}(\hat{\bm{\lambda}},\bar{\bm{y}} )}$, where $\alpha$ denotes the step size. Plug $\tilde{\bm{\lambda}}$ into \eqref{eq: convergence1}, 

\begin{equation} \label{eq: convergence2}
\begin{aligned}
&{ \mathcal{L}_\epsilon^{PFY}(\tilde{\bm{\lambda}},\bar{\bm{y}} )} 
 \\
 &\leq { \mathcal{L}_\epsilon^{PFY}(\hat{\bm{\lambda}},\bar{\bm{y}} )} +\nabla { \mathcal{L}_\epsilon^{PFY}(\hat{\bm{\lambda}},\bar{\bm{y}} )}^T (\hat{\bm{\lambda}}- \alpha { \nabla \mathcal{L}_\epsilon^{PFY}(\hat{\bm{\lambda}},\bar{\bm{y}} )} -\hat{\bm{\lambda}})\\&+\frac{\mathcal{R}_{\mathcal{X}}M}{2\epsilon}||\hat{\bm{\lambda}}- \alpha \nabla { \mathcal{L}_\epsilon^{PFY}(\hat{\bm{\lambda}},\bar{\bm{y}} )} -\hat{{\bm{\lambda}}} ||_2^2 \\
&=
{ \mathcal{L}_\epsilon^{PFY}(\hat{\bm{\lambda}},\bar{\bm{y}} )} -\nabla { \mathcal{L}_\epsilon^{PFY}(\hat{\bm{\lambda}},\bar{\bm{y}} )}^T( \alpha { \nabla \mathcal{L}_\epsilon^{PFY}(\hat{\bm{\lambda}},\bar{\bm{y}} )})+\\&\frac{\mathcal{R}_{\mathcal{X}}M \alpha^2}{2\epsilon}|| \nabla { \mathcal{L}_\epsilon^{PFY}(\hat{\bm{\lambda}},\bar{\bm{y}} )}  ||_2^2 \\
&= 
{ \mathcal{L}_\epsilon^{PFY}(\hat{\bm{\lambda}},\bar{\bm{y}} )} - \alpha || \nabla { \mathcal{L}_\epsilon^{PFY}(\hat{\bm{\lambda}},\bar{\bm{y}} )}  ||_2^2+\\
&\frac{\mathcal{R}_{\mathcal{X}}M \alpha^2}{2\epsilon}|| \nabla { \mathcal{L}_\epsilon^{PFY}(\hat{\bm{\lambda}},\bar{\bm{y}} )}  ||_2^2 \\
&= 
 \mathcal{L}_\epsilon^{PFY}(\hat{\bm{\lambda}},\bar{\bm{y}} ) - \alpha (1-\frac{\mathcal{R}_{\mathcal{X}}M \alpha}{2\epsilon})|| \nabla { \mathcal{L}_\epsilon^{PFY}(\hat{\bm{\lambda}},\bar{\bm{y}} )}  ||_2^2 \\
\end{aligned}
\end{equation}

\noindent When $\alpha \leq \frac{\epsilon}{\mathcal{R}_{\mathcal{X}}M}$, $-(1-\frac{\mathcal{R}_{\mathcal{X}}M \alpha}{2\epsilon}) =(\frac{\mathcal{R}_{\mathcal{X}}M \alpha}{2\epsilon}-1) \leq -\frac{1}{2}$

\begin{equation} \label{eq: convergence3}
\begin{aligned}
&{ \mathcal{L}_\epsilon^{PFY}(\tilde{\bm{\lambda}},\bar{\bm{y}} )} 
\leq
 \mathcal{L}_\epsilon^{PFY}(\hat{\bm{\lambda}},\bar{\bm{y}} ) - \frac{\alpha}{2}|| \nabla { \mathcal{L}_\epsilon^{PFY}(\hat{\bm{\lambda}},\bar{\bm{y}} )}  ||_2^2 \\
\end{aligned}
\end{equation}

Inequality \eqref{eq: convergence3} implies that at each iteration, the value of the objective function decreases until it reaches the optimal value ${\mathcal{L}_\epsilon^{PFY}(\hat{\bm{\lambda}}^*,\bar{\bm{y}} )}=0$.

Since ${\mathcal{L}_\epsilon^{PFY}(\hat{\bm{\lambda}},\bar{\bm{y}} )}$ is convex,

\begin{equation} \label{eq: convergence4}
\begin{aligned}
&{ \mathcal{L}_\epsilon^{PFY}(\hat{\bm{\lambda}}^*,\bar{\bm{y}} )} 
\geq 
 \mathcal{L}_\epsilon^{PFY}(\hat{\bm{\lambda}},\bar{\bm{y}} ) + \nabla \mathcal{L}_\epsilon^{PFY}(\hat{\bm{\lambda}},\bar{\bm{y}} )^T(\bm{\lambda}^*-\bm{\lambda}) \\
\end{aligned}
\end{equation}
Move ${ \mathcal{L}_\epsilon^{PFY}(\hat{\bm{\lambda}},\bar{\bm{y}} )} $ to the left-hand side,
\begin{equation} \label{eq: convergence5}
\begin{aligned}
&{ \mathcal{L}_\epsilon^{PFY}(\hat{\bm{\lambda}},\bar{\bm{y}} )} 
\leq 
 \mathcal{L}_\epsilon^{PFY}(\hat{\bm{\lambda}}^*,\bar{\bm{y}} ) + \nabla \mathcal{L}_\epsilon^{PFY}(\hat{\bm{\lambda}},\bar{\bm{y}} )^T(\bm{\lambda}-\bm{\lambda}^*) \\
\end{aligned}
\end{equation}

\noindent Plug \eqref{eq: convergence5} into \eqref{eq: convergence3}, we obtain

\begin{equation} \label{eq: convergence6}
\begin{aligned}
{\mathcal{L}_\epsilon^{PFY}(\tilde{\bm{\lambda}},\bar{\bm{y}})} 
\leq \, \, &
 \mathcal{L}_\epsilon^{PFY}(\hat{\bm{\lambda}}^*,\bar{\bm{y}} ) + \nabla \mathcal{L}_\epsilon^{PFY}(\hat{\bm{\lambda}},\bar{\bm{y}} )^T(\bm{\lambda}-\bm{\lambda}^*) \\&- \frac{\alpha}{2}|| \nabla { \mathcal{L}_\epsilon^{PFY}(\hat{\bm{\lambda}},\bar{\bm{y}} )}  ||_2^2 \\
\end{aligned}
\end{equation}
Thus,

\begin{equation} \label{eq: convergence7}
\begin{aligned}
&{ \mathcal{L}_\epsilon^{PFY}(\tilde{\bm{\lambda}},\bar{\bm{y}} )} - \mathcal{L}_\epsilon^{PFY}(\hat{\bm{\lambda}}^*,\bar{\bm{y}} ) 
\\
&\leq  \nabla \mathcal{L}_\epsilon^{PFY}(\hat{\bm{\lambda}},\bar{\bm{y}} )^T(\bm{\lambda}-\bm{\lambda}^*) - \frac{\alpha}{2}|| \nabla { \mathcal{L}_\epsilon^{PFY}(\hat{\bm{\lambda}},\bar{\bm{y}} )}  ||_2^2 \\
&=  \frac{1}{2\alpha} ({2\alpha} \nabla \mathcal{L}_\epsilon^{PFY}(\hat{\bm{\lambda}},\bar{\bm{y}} )^T(\bm{\lambda}-\bm{\lambda}^*) - {\alpha}^2|| \nabla { \mathcal{L}_\epsilon^{PFY}(\hat{\bm{\lambda}},\bar{\bm{y}} )}  ||_2^2) \\
&=  \frac{1}{2\alpha} ({2\alpha} \nabla \mathcal{L}_\epsilon^{PFY}(\hat{\bm{\lambda}},\bar{\bm{y}} )^T(\bm{\lambda}-\bm{\lambda}^*) - {\alpha}^2|| \nabla { \mathcal{L}_\epsilon^{PFY}(\hat{\bm{\lambda}},\bar{\bm{y}} )}||_2^2 \\& -||\hat{\bm{\lambda}}-\hat{\bm{\lambda}}^*||_2^2+||\hat{\bm{\lambda}}-\hat{\bm{\lambda}}^*||_2^2) \\
&=  \frac{1}{2\alpha} (||\hat{\bm{\lambda}}-\hat{\bm{\lambda}}^*||_2^2-||\hat{\bm{\lambda}}-\nabla { \mathcal{L}_\epsilon^{PFY}(\hat{\bm{\lambda}},\bar{\bm{y}} )}-\hat{\bm{\lambda}}^*||_2^2) \\
&=  \frac{1}{2\alpha} (||\hat{\bm{\lambda}}-\hat{\bm{\lambda}}^*||_2^2-||\tilde{\bm{\lambda}}-\hat{\bm{\lambda}}^*||_2^2) \\
\end{aligned}
\end{equation}

The inequality in \eqref{eq: convergence7} holds for every iteration of gradient descent.  By replacing $\tilde{\bm{\lambda}}^*$ with $\hat{\bm{\lambda}}^{(i)}$ for $k$ iterations and summing them together, we obtain

\begin{equation} \label{eq: convergence8}
\begin{aligned}
&\sum_{i=1}^k{ (\mathcal{L}_\epsilon^{PFY}(\hat{\bm{\lambda}}^{(i)},\bar{\bm{y}} )} - \mathcal{L}_\epsilon^{PFY}(\hat{\bm{\lambda}}^*,\bar{\bm{y}}) ) 
\\
&= \sum_{i=1}^k \frac{1}{2\alpha} (||\hat{\bm{\lambda}}^{(i-1)}-\hat{\bm{\lambda}}^*||_2^2-||\hat{\bm{\lambda}}^{(i)}-\hat{\bm{\lambda}}^*||_2^2) \\
&= \frac{1}{2\alpha} (||\hat{\bm{\lambda}}^{(0)}-\hat{\bm{\lambda}}^*||_2^2-||\hat{\bm{\lambda}}^{(k)}-\hat{\bm{\lambda}}^*||_2^2) \\
&\leq   \frac{1}{2\alpha} (||\hat{\bm{\lambda}}^{(0)}-\hat{\bm{\lambda}}^*||_2^2) \\
\end{aligned}
\end{equation}
where the transition from the first equality to the second equality  is based on the telescopic sum. Because the value of $\mathcal{L}_\epsilon^{PFY}(\hat{\bm{\lambda}},\bar{\bm{y}} )$ is decreasing in each iteration until it reaches the optimal value, we have

% \begin{equation} \label{eq: convergence7}
% \begin{aligned}
% &k{( \mathcal{L}_\epsilon^{PFY}(\hat{\bm{\lambda}}^{(k)},\bar{\bm{y}} ) } - \mathcal{L}_\epsilon^{PFY}(\hat{\bm{\lambda}}^*,\bar{\bm{y}} ) ) 
% \leq \sum_{i=1}^k{ (\mathcal{L}_\epsilon^{PFY}(\hat{\bm{\lambda}}^{(i)},\bar{\bm{y}} )} -\\&
% \mathcal{L}_\epsilon^{PFY}(\hat{\bm{\lambda}}^*,\bar{\bm{y}} ) )
% \end{aligned}
% \end{equation}

\begin{equation} \label{eq: convergence8}
\begin{aligned}
&{ \mathcal{L}_\epsilon^{PFY}(\hat{\bm{\lambda}}^{(k)},\bar{\bm{y}} ) } - \mathcal{L}_\epsilon^{PFY}(\hat{\bm{\lambda}}^*,\bar{\bm{y}} )\\ 
&\leq \frac{1}{k}\sum_{i=1}^k{ (\mathcal{L}_\epsilon^{PFY}(\hat{\bm{\lambda}}^{(i)},\bar{\bm{y}} )} -
\mathcal{L}_\epsilon^{PFY}(\hat{\bm{\lambda}}^*,\bar{\bm{y}} ) ) \\
& \leq  \frac{||\hat{\bm{\lambda}}^{(0)}-\hat{\bm{\lambda}}^*||_2^2}{2\alpha k} 
\end{aligned}
\end{equation}

Therefore, we conclude the proof.

\end{proof} 

\end{document}